\documentclass{article}
\usepackage{algorithm}
\usepackage{algpseudocode} 
\usepackage{amsmath}
\usepackage{amssymb}
\usepackage{amsthm}
\usepackage{bbm}
\usepackage{booktabs}   
\usepackage[a4paper]{geometry}
\usepackage{xcolor}
\usepackage{subcaption}
\usepackage{hyperref}
\usepackage{graphicx}
\usepackage[numbers,sort&compress]{natbib}
\usepackage{authblk}   

\usepackage{tikz}
\tikzset{node distance=4cm, auto}
\usetikzlibrary{decorations.markings}

\newtheorem{proposition}{Proposition}[section]
\newtheorem{lemma}[proposition]{Lemma}
\newtheorem{theorem}[proposition]{Theorem}
\newtheorem{corollary}[proposition]{Corollary}

\theoremstyle{definition}

\newtheorem{assumption}[proposition]{Assumption}

\theoremstyle{remark}
\newtheorem{remark}{Remark}[section]

\newcommand{\dif}{\mathrm{d}}
\newcommand{\sgn}{\operatorname{sgn}}

\newcommand{\E}{\mathbb E}
\renewcommand{\P}{\mathbb P}

\newcommand{\Kclipp}{K_{\text{clipp}}}
\newcommand{\Ksnapp}{K_{\text{snapp}}}

\title{Transient regime of piecewise deterministic Monte Carlo algorithms}

\author{Sanket Agrawal\thanks{Department of Statistics, University of Warwick, UK.
Email: \href{mailto:sanket.Agrawal@warwick.ac.uk}{sanket.Agrawal@warwick.ac.uk}}, 
Joris Bierkens\thanks{Delft Institute of Applied Mathematics, TU Delft, NL.
Email: \href{mailto:joris.bierkens@tudelft.nl}{joris.bierkens@tudelft.nl}},
Kengo Kamatani\thanks{Institute of Statistical Mathematics, Tokyo, JP.
Email: \href{mailto:kamatani@ism.ac.jp}{kamatani@ism.ac.jp}}, 
and Gareth O.\ Roberts\thanks{Department of Statistics, University of Warwick, UK.
Email: \href{mailto:gareth.O.Roberts@warwick.ac.uk}{gareth.O.Roberts@warwick.ac.uk}}}

\date{}

\begin{document}
\maketitle

\begin{abstract}
Piecewise Deterministic Markov Processes (PDMPs) such as the Bouncy Particle Sampler and the Zig-Zag Sampler, have gained attention as continuous-time counterparts of classical Markov chain Monte Carlo. We study their transient regime under convex potentials, namely how trajectories that start in low-probability regions move toward higher-probability sets. Using fluid-limit arguments with a decomposition of the generator into fast and slow parts, we obtain deterministic ordinary differential equation descriptions of early-stage behaviour. The fast dynamics alone are non-ergodic because once the event rate reaches zero it does not restart. The slow component reactivates the dynamics, so averaging remains valid when taken over short micro-cycles rather than with respect to an invariant law.

Using the expected number of jump events as a cost proxy for gradient evaluations, we find that for Gaussian targets the transient cost of PDMP methods is comparable to that of random-walk Metropolis. For convex heavy-tailed families with subquadratic growth, PDMP methods can be more efficient when event simulation is implemented well. Forward Event-Chain and Coordinate Samplers can, under the same assumptions, reach the typical set with an order-one expected number of jumps. For the Zig-Zag Sampler we show that, under a diagonal-dominance condition, the transient choice of direction coincides with the solution of a box-constrained quadratic program; outside that regime we give a formal derivation and a piecewise-smooth update rule that clarifies the roles of the gradient and the Hessian. These results provide theoretical insight and practical guidance for the use of PDMP samplers in large-scale inference.
\end{abstract}



\section{Introduction}

Piecewise Deterministic Markov Processes (PDMPs), such as the Bouncy Particle Sampler and  the Zig-Zag Sampler, have recently attracted significant attention as alternatives to traditional Markov chain Monte Carlo (MCMC) methods \cite{PhysRevE.85.026703, Bouchard2018, Bierkens2019}. These samplers are appealing for large-scale Bayesian computation due to their non-reversible dynamics and their potential to explore high-dimensional targets more efficiently \cite{andrieu2021hypocoercivity, bierkens2018high, 10.3150/24-BEJ1807}.

In this study, we investigate the transient regime of PDMP-based samplers under convex target potentials. We focus on the transient regime, namely the behaviour of processes initialised in low-probability regions and how they move toward higher-probability sets. Understanding this transition enables a direct comparison with conventional MCMC methods in terms of their convergence rates. Our approach is rooted in the framework of fluid limits, where Markov processes are approximated by ordinary differential equations in an appropriate scaling regime.

The study of fluid limits in Markov chains dates back to \cite{Feller1939} which studies the convergence of a birth-death process to the Malthus model. Later, a unified theoretical framework was developed by \cite{Kurtz1970}. Since then, the approach has been applied to various stochastic systems, such as queueing networks, population dynamics, chemical kinetics, and epidemiological models. This approach has also been applied to investigating the performance of MCMC. In particular, \cite{Christensen2005} showed that in high-dimensional Metropolis-Hastings, the sequence of potential-function values initially evolves along a nearly deterministic trajectory, with stochastic fluctuations becoming significant only after this transient regime. Further work by \cite{Fort2008} applied the fluid limit analysis not only to the transient regime but also to long-term ergodic behaviour. 

Recent advances in high-dimensional asymptotics have further connected Monte Carlo dynamics with fluid limits.  For instance, \cite{Christensen2005, jourdain2014optimal, jourdain2015meanfield} rigorously characterised the transient regime of both the random-walk Metropolis and the Langevin algorithms as the dimension tends to infinity.  More recently, \cite{kuntz2018nonstationary} extended these ideas to the non-stationary phase of the Langevin algorithm under Gaussian perturbations in an infinite-dimensional Hilbert-space setting.
Other related works include
\cite{Beskos2013}, which analysed the scaling limit of Hamiltonian Monte Carlo toward randomised Hamiltonian dynamics, and \cite{Deligiannidis2021}, which studied the convergence properties of the Bouncy Particle Sampler to the same class of processes. Although these papers did not focus on the transient regime, they highlight the relevance of fluid limit analysis in understanding the scaling properties and efficiency of modern sampling algorithms.

Our work belongs to this line of research and extends fluid-limit methods to study the transient behaviour of PDMP-based samplers. 
In this work, we focus on the fixed-dimensional transient regime, following \cite{Fort2008}, in which both the coordinate process  and the potential‐function values evolve deterministically according to an ODE possibly with jumps. This regime sits between the stationary  regime, where both levels exhibit stochastic fluctuations, and the high-dimensional fluid-limit regime, where the potential‐function trajectories become deterministic but the underlying coordinate dynamics remain stochastic. 

From a theoretical standpoint, we analyse the fluid limit by the averaging principle for Markov processes. Concretely, 
for a family of Markov processes with generators $\{\mathcal{L}^\varepsilon  \}$ indexed by a continuous parameter $\varepsilon >0$,
we decompose the generator as $\mathcal{L}^\varepsilon = \varepsilon^{-1} \mathcal{L}_0 + \mathcal{L}_1$ with $\varepsilon \to 0$. The operator $\mathcal{L}_0$ of an ergodic Markov process governs fast dynamics that are averaged out, while the remaining slow dynamics are captured by $\mathcal{L}_1$ assuming that some components are averaged out by the invariant distribution associated with $\mathcal{L}_0$. This approach, which can be traced back at least to \cite{khasminskii1968averaging}, has been widely studied in the literature, including in numerous Monte Carlo applications such as \cite{Sherlock2015, beskos2018, Vandecasteele2023}. 

As an application of the averaging principle,  the problems considered in this paper are `fully-coupled'  in the sense that the fast dynamics also depends on the slow dynamics. 
Also, our problem deviates from the standard averaging principle because the leading operator $\mathcal{L}_0$, which governs the fast dynamics, does not generate  an ergodic Markov process. In particular, $\mathcal{L}_0$ can drive the intensity function to zero, but cannot restore it once it is zero. By contrast, the slow-dynamics operator $\mathcal{L}_1$ can `turn on' the intensity again, thereby compensating for $\mathcal{L}_0$’s lack of ergodicity. As a result, despite $\mathcal{L}_0$ not being ergodic on its own, the full system still allows for a modified version of the usual averaging argument, owing to the properties of the slow dynamics.

This compensatory effect arises from the convexity of the potential function. Consequently, for a non-convex target distribution, the averaging principle does not adequately describe the limiting dynamics, and the process can reach higher-probability regions with relatively few jumps. See Remarks \ref{remark:non-convex-bps} and \ref{remark:non-convex-zz} for further details.

Our analysis leads to the following observations on computational efficiency. Throughout, we measure cost by the number of jump events required to reach the high-probability region, since, under an ideal implementation, the expected jump count is, up to a constant factor, proportional to the number of gradient evaluations in simulation (see Subsection~\ref{subsec:comp-cost}).

\begin{itemize}
    \item[(i)] 
For Gaussian target distributions, the computational efficiency of the Bouncy Particle Sampler and the Zig-Zag Sampler is comparable to that of the random-walk Metropolis algorithm. Notably, when the tails of the target distribution are heavier than Gaussian, the PDMP-based samplers outperform random-walk Metropolis if implementing these methods efficiently, such as in \cite{Sutton02102023, andral2024automated}. 
\item[(ii)] 
Moreover, the Bouncy Particle Sampler becomes more efficient when the refreshment rate is properly tuned. In this algorithm, there are two types of jump events: bounces, where the velocity is reflected at level sets of the target density, and refreshments, where the velocity is randomly resampled to maintain ergodicity. The performance improves when bounce and refreshment events occur at similar frequencies. This balance‐of‐events criterion is also suggested in \cite{bierkens2018high}. 
\item[(iii)] Additionally, we find that employing the Forward Event-Chain Sampler \cite{Michel2017} and the Coordinate Sampler \cite{Wu2020CoordinateSampler}, significantly improves computational efficiency, achieving an $O(1)$ number of jump events to enter the central (high-density) region of the target. Our proof proceeds by establishing an exponential drift condition in which the usual `small set' is replaced by this central region, thereby guaranteeing rapid approach to the typical set.
\item[(iv)] 
Lastly, we observe that during the transient phase, the direction of motion in the Zig-Zag Sampler is determined by solving a constrained optimisation problem. More precisely, the process follows a smooth path until it reaches a boundary where the gradient of the potential is zero, at which point it selects its direction by solving an optimisation problem governed by the Hessian and the constraint $|v| \leq 1$. 
\end{itemize}

Our paper is structured as follows. In Section~\ref{sec:pdmp}, we present the averaging principle that underpins all subsequent results. Section~\ref{sec:fluid_bps} investigates Bouncy Particle-family samplers in three regimes. 
In its low-refreshment regime, these samplers' dynamics are dominated by  frequent  rapid reversals of direction causing the trajectory to hug a curve which consequently becomes the fluid limit. We call this phenomenon {\em snapping}.
Furthermore,  in a high-refreshment regime, we observe a reduction  but not an entire elimination of  snapping.
However, moving to a randomised bounce regime (i.e., the Forward Event-Chain Sampler), we observe behaviour which  entirely avoids snapping and therefore becomes markedly more efficient than the first two cases. 
Section~\ref{sec:zzs_averaging} begins with the optimisation problem that characterises the limiting direction choice of the Zig-Zag Sampler and then derives its fluid limit in the simplest setting. Although the Zig-Zag Sampler can still exhibit snapping, we show that its close relative, the Coordinate Sampler, avoids this issue. Section~\ref{sec:discussion} concludes with a discussion of practical implications and open questions.
Technical details are deferred to Supplementary Material. The full proof of the averaging principle appears in Appendix~\ref{app:sec:proof_PDMP}; a complementary analysis of the number of jumps required to reach the high-probability region is given in Appendix~\ref{app:sec:fec}. 

\subsection{Setting for the fluid limit}
\label{subsec:setting}

Let $U:\mathbb{R}^N \to \mathbb{R}$ be a given potential function such that $\int \exp(-U(x))\ \dif x < \infty$. We consider a family of target distributions $\{\Pi^{\varepsilon}, \varepsilon > 0\}$ where 
\[
    \Pi^{\varepsilon}(\dif x) \propto \exp\left(-\varepsilon^{-1}U(x)\right), \quad x \in \mathbb{R}^N. 
\]
Let $\Phi_t(x)$ be a PDMP-based sampler targeting $\Pi(\dif x) \propto \exp(-U(x))\dif x$ with initial condition $\Phi_0(x)=x$. We are then interested in the weak limit of $\Phi^{\varepsilon}_t$ where $\Phi^{\varepsilon}_t$ is the same PDMP-based sampler with target distribution $\Pi^{\varepsilon}$ as $\varepsilon \to 0$. 

Note that this is different from the classical fluid limit framework of \cite{Fort2008} where, given $\Phi_t(x)$ with $\Phi_0(x) = x$, they establish, for some $\alpha \ge 0$ and all $x \in \mathbb{R}^N$, a weak limit for the rescaled process,
\[
    \Phi^{\varepsilon}_{t, \alpha} := \varepsilon \Phi_{t\varepsilon^{-1-\alpha}}(x/\varepsilon),
\]
as $\varepsilon \to 0$. 
Any resulting limit is referred to as the \emph{fluid limit} of $\Phi_t$ and, in general, is a weak solution to an ODE. The same procedure carries over to discrete-time Markov chains by embedding them in continuous time via $t \mapsto \Phi_{\lfloor t \rfloor}(x)$, where $\lfloor t \rfloor$ denotes the greatest integer less than or equal to $t$. 


Our motivation for rescaling the potential function instead of space and time (as in \cite{Fort2008}) stems from the fact that it induces a difference in time-scales of the components of the underlying stochastic process. This allows us to establish fluid limits via the averaging principle (described in Section \ref{subsec:averaging-pdmp}) for Markov processes. Although mathematically motivated, this situation also arises naturally in practice, e.g.,  when targeting a Bayesian posterior for large datasets, $\varepsilon = n^{-1}$ where $n$ is the dataset size (\cite{agrawal2024large}). Moreover, we emphasise that this potential-rescaling view remains compatible with the classical fluid-limit program of \cite{Fort2008}; for example, consider the power-exponential family
\[
    \Pi(\dif x) \;\propto\; \exp\bigl(-\bigl(x^\top \Sigma^{-1} x/2\bigr)^{\beta/2}\bigr)\,\dif x
\]
for $\beta > 1$. Then, rescaling the potential by 
$\varepsilon^{-1}$ is equivalent to the space-time scaling, with $x\mapsto r^{-1}x$ and $\varepsilon=r^{-\beta}$ and the two limits coincide.
In particular, the regime $1< \beta<2$ corresponds to subquadratic growth, whereas $\beta>2$ gives superquadratic growth.

\subsubsection*{Fluid limit of the random walk Metropolis chains}

The fluid limit of a random walk Metropolis algorithm was studied in \cite{Fort2008}. Here we illustrate their result in a simple case. 
Specifically, for $\beta>1$ in the above setting, we examine the potential function $\varepsilon^{-1}U(x)$ and the increment distribution $\mathcal{N}(0, \varepsilon^{2/\beta}\sigma^2 I_N)$ with state space scaling $x\mapsto \varepsilon^{1/\beta}x$ associated with the  algorithm. 


A function is called coercive if its all sublevel sets are relatively compact. If the function $U$ is coercive, and twice continuously differentiable, then the fluid limit of the Metropolis algorithm with $\alpha=0$ is the following normalised gradient flow:  
\[
   \dot{x}(t) = -\,\frac{\sigma}{\sqrt{2\pi}}\frac{\nabla U\bigl(x(t)\bigr)}{|\nabla U\bigl(x(t)\bigr)|}
\]
until the process hits $\{x:U(x)\le U(x^*)+\gamma\}$ for $\gamma>0$ where $x^*$ is the global minimiser of $U$. 
See Proposition 2.6 of \cite{Fort2008} for further details. 


Due to the time scaling $\varepsilon^{-1/\beta}$, the Metropolis algorithm’s computational cost to reach $\{x : U(x)\le U(x^*)+\gamma\}$ is of the order of $\varepsilon^{-1/\beta}$ in this setting. 

\begin{remark}
In this paper, we do not consider gradient-based Metropolis algorithms, as their performance is highly sensitive to both the structure of the potential function and the choice of tuning parameters (see Lemma 3 of \cite{Christensen2005} and Section 6.1.3 of \cite{jourdain2014optimal}). This is in contrast to PDMP-based samplers, where gradients influence only the intensity function, leading to greater robustness.  Methods based on the Barker acceptance have better robustness \cite{zanella2020informed}, but in our framework, they sacrifice the efficiency gains one achieves by employing the drift coefficient.
\end{remark}

\subsection{Problem setup}\label{subsec:setup}

Throughout this paper, we work under a set of technical assumptions on the potential function $U$, which defines the stationary distribution of the Piecewise Deterministic Markov Process (PDMP) under consideration. Our primary goal is to understand how the state process of the PDMP-based samplers converges, in a suitable sense, to a deterministic flow, particularly as it approaches the global minimiser of the potential, denoted by $x^*$.

\begin{assumption}\label{assumption}
The potential function $U$ is continuously differentiable on $\mathbb{R}^N$ and three times continuously differentiable away from the minimiser, i.e., $U \in C^1(\mathbb{R}^N) \cap C^3(\mathbb{R}^N \setminus \{x^*\})$, where $x^*$ is the unique point satisfying $\nabla U(x^*) = 0$. In addition, the Hessian matrix $\nabla^2 U(x)$ is positive definite for all $x \neq x^*$, and $U(x) \to \infty$ as $|x| \to \infty$.
\end{assumption}

Our main result establishes the convergence of PDMP-based samplers to a deterministic flow $\phi_t(x)$ (or, $z$ in place of $x$), up to a hitting time defined by
\[
\tau(x; \gamma) = \inf\left\{t > 0 : U(\phi_t(x)) \leq U(x^*) + \gamma \right\},
\]
assuming the initial point satisfies $U(x) > U(x^*) + \gamma$ for $\gamma>0$. 
Note that $\tau(x; \gamma)=+\infty$ for the limit of the Bouncy Particle Sampler with low refreshment jumps and we show the convergence of $X_t^\varepsilon$ up to a fixed time $T>0$. 

In a complementary direction, we also study the number of velocity refreshments required for the process to approach the minimiser. Let $T_1 < T_2 < \cdots$ denote the refreshment times of the velocity component. We investigate the tightness of the stopping index
\[
n^\varepsilon(x^\varepsilon; \gamma) = \inf\left\{n \in \mathbb{N} : U(X_{T_n}^\varepsilon) \leq U(x^*) + \gamma \right\},
\]
when the process starts from $x^\varepsilon$, 
showing that the process typically enters a neighbourhood of the minimiser after a finite number of refreshment jumps. 
This provides an information about the overall number of jumps until the process hits the high probability region. 

Our main results rely on an averaging principle described in Section \ref{subsec:averaging-pdmp}, while the analysis in this complementary direction makes use of a drift inequality which will be described in  Appendix~\ref{app:sec:fec}.

\section{Piecewise Deterministic Markov processes}\label{sec:pdmp}

\subsection{Markov processes and their extended generators}

As explained in the Introduction, we consider the fluid limit of several piecewise deterministic Monte Carlo algorithms, specifically, the Bouncy Particle Sampler, the Forward Event-Chain, the Coordinate Sampler and the Zig-Zag Sampler. We first introduce the extended generator of these piecewise deterministic Monte Carlo algorithms.

Suppose that the target distribution $\Pi$ on $\mathbb{R}^N$ has a density function $\exp(-U(x))$ with respect to Lebesgue measure. 
Let $x$ be the state variable and $v$ be the velocity variable. The extended generator of the Bouncy Particle Sampler is given by
\[
\mathcal{L}_{\text{BPS}} f(x,v) = v^{\top}\nabla_x f(x,v)+ (v^{\top}\nabla U(x))_+\bigl(\mathcal{B} - I\bigr) f(x,v)  
\]
where $\mathcal{B} f(x,v) = f\bigl(x, v - 2 n(x) n(x)^\top v\bigr)$ and for $x\neq x^*$, 
\begin{align*}    
n(x) = \frac{\nabla U(x)}{|\nabla U(x)|},\quad 
n^\perp(x)=(I-n(x)n(x)^{\top}). 
\end{align*}
To prevent reducibility of the Markov process, a refreshment mechanism is typically introduced. At refreshment events, occurring according to a Poisson process with a fixed refreshment rate, the particle’s velocity is randomly resampled, ensuring ergodicity. However, for simplicity, we omit this term from the displayed generator. 

A variant of the process,  the Forward Event-Chain, was proposed by \cite{vanetti2017piecewise} and \cite{Michel2017}. It replaces  the deterministic operator $\mathcal{B}$ by the Markov transition kernel $\mathcal{Q}$:
\begin{equation}\label{eq:fec}
\mathcal{Q}f(x,v)=\int_{\mathbb{R}^N}\int_0^\infty f(x,-\xi n(x)+n^\perp(x)w)R(\dif \xi)\varphi(w)\dif w
\end{equation}
where $R(\dif \xi)=\xi\exp(-\xi^2/2)1_{[0,\infty)}(\xi)\dif \xi$ is the Rayleigh$(1)$ distribution, and $\varphi(w)$ is the density of the standard normal distribution on $\mathbb{R}^N$.  Although this modification appears minor when compared with the Bouncy Particle Sampler, we shall show that it fundamentally alters the limit behaviour of the process.

Another variant is the Coordinate Sampler proposed by \cite{Wu2020CoordinateSampler}. The space of $v$ is $V:=\{\pm e_n: n=1,\ldots, N\}$ where $e_n$ denotes the 
$n$-th canonical basis vector in $\mathbb{R}^N$. It replaces $\mathcal{B}$ by operation $\mathcal{C}$: 
$$
\mathcal{C}f(x,v)=\sum_{v'\in V}~\frac{(v'^{\top}\nabla U(x))_+}{\sum_{n=1}^N|\partial_nU(x)|}~f(x, -v')
$$
where $\partial_nU(x)=\partial U(x)/\partial x_n$.

The Zig-Zag Sampler proposed by \cite{Bierkens2019} is another class of PDMP-based samplers, where the velocity $v$ takes values in the discrete set $\{-1,+1\}^N$. The sampler explores the state space by continuously switching directions along coordinate axes, producing a zig-zag trajectory. 
The extended generator of the Zig-Zag Sampler is given by
\[
\mathcal{L}_{\text{ZZS}} f(x,v) = v^{\top}\nabla_x f(x,v)+ \sum_{n=1}^N(v_n\partial_nU(x))_+\bigl(\mathcal{F}_n - I\bigr) f(x,v),
\]
with $\mathcal{F}_n f(x,v) = f\bigl(x,F_n(v)\bigr)$ where $F_n$ is the operation that switches the sign of the $n$-th coordinate of $v$. The space of $v$ is restricted to $\{-1,+1\}^N$ for the Zig-Zag Sampler. 

The first component of each extended generator is the drift term, corresponding to $\dot x=v$; the second is the jump part. The jump is triggered by the intensity function $(v^{\top}\nabla U(x))_+$ and the velocity switches from $v$ to $v - 2 n(x) n(x)^{\top} v$ deterministically in the Bouncy Particle Sampler. The Forward Event-Chain replaces the deterministic update by a random update. At each event the Coordinate Sampler updates only one component: it selects index $n$ with probability proportional to $|\partial_n U(x)|$, then sets the direction to $-\sgn(\partial_n U(x))\,e_n$.
For the Zig-Zag Sampler, the jump is triggered by the series of intensity functions $(v_n\partial_n U(x))_+$ for $n=1,\ldots N$, and the velocity switches from $v$ to $F_n(v)$ corresponding to source of the jump. 

For the fluid limit, we are interested in studying its rescaled generator $\mathcal{L}^\varepsilon$. In this case, the potential function  $U(x)$ is substituted by  $\varepsilon^{-1}U(x)$. For the analysis, we consider the averaging of Markov processes described in Subsection \ref{subsec:averaging-pdmp}. 

\begin{remark}
The time-scaling exponent $\alpha = -1$ introduced in Section \ref{subsec:setting} is not written explicitly for PDMP samplers, yet it is implicitly assumed: rescaling the initial condition in this framework automatically accelerates the process.
\end{remark}

\begin{remark}
The Forward Event-Chain literature includes several variants; see \cite{vanetti2017piecewise,Michel2017}.  Our basic update with any of these alternatives leaves all main results essentially unchanged.  We nevertheless adopt the present ``orthogonal refresh'' because it turns the jump kernel into a genuine Markov chain, making the expression simpler.  For the rest of the paper we therefore focus on this  kernel.
\end{remark}

\subsection{Averaging of piecewise deterministic Markov processes}
\label{subsec:averaging-pdmp}

We briefly explain the averaging principle.
Consider a Markov process with state $(\mathbf{x},\mathbf{y})\in\mathbb{R}^{p+q}$,
where $\mathbf{x}$ collects the \emph{slow} variables and $\mathbf{y}$ the \emph{fast} ones, as summarised below: 
\[
\renewcommand{\arraystretch}{1.05}
\begin{array}{c|c}
\textbf{Slow variables } & \textbf{Fast variables } \\ \hline
\mathbf{x} & \mathbf{y}
\end{array}
\]
The extended generator then admits the expansion
$$
\mathcal{L}^\varepsilon=\varepsilon^{-1}\mathcal{L}_0+\mathcal{L}_1+O(\varepsilon). 
$$
To analyse the limit behaviour of the process, we consider the following backward Kolmogorov equation for $u^\varepsilon(\mathbf{x},\mathbf{y},t)$:
\begin{align*}
    \frac{\dif u^\varepsilon}{\dif t}=\mathcal{L}^\varepsilon u^\varepsilon. 
    \end{align*}
With the formal expansion $u^\varepsilon=u_0+\varepsilon u_1+O(\varepsilon^2)$, we have microscopic and macroscopic equations: 
\begin{align*}
\begin{cases}
0&=\mathcal{L}_0 u_0,\\
\dfrac{\dif u_0}{\dif t}&=\mathcal{L}_0 u_1+\mathcal{L}_1 u_0.
\end{cases}
\end{align*}
Suppose that $\mathcal{L}_0$ is an operator for an ergodic Markov process acting on the $\mathbf{y}$-component , with invariant measure $\mu$.
 Then, since the microscopic,  Poisson equation implies that $u_0$ is invariant under the action of $\mathbf{y}$, it must be a constant. In other words, $u_0$ does not depend on $\mathbf{y}$, and we may write $u_0(\mathbf{x},\mathbf{y},t)=u_0(\mathbf{x},t)$. 
Then, by integrating out $\mathbf{y}$ with respect to $\mu$, the macroscopic equation becomes the effective equation: 
$$
\frac{\dif u_0(\mathbf{x},t)}{\dif t}=\mu\left(\mathcal{L}_1 u_0(\mathbf{x},t)\right).
$$
This equation indicates the limit of the process. This is the standard scenario in averaging for Markov processes.

In our setting the fast operator $\mathcal L_0$ is \emph{not} ergodic, so we do not average against an invariant law. 
Instead we average over micro-cycles between successive zero-crossings of 
\[
\lambda(x,v):=v^\top\nabla U(x)
\]
for the Bouncy Particle Sampler. 
The fast part flips whenever $\lambda>0$ but cannot flip for negative values; convexity (Assumption \ref{assumption}) makes $\dot\lambda=v^\top\nabla^2 U(x)\,v> 0$ on continuity pieces, pushing $\lambda$ back to zero and closing the cycle. 
Cycle-averaging the microscopic update of the tangential velocity $v_1:=n^\perp(x)v$ then we obtain an effective ordinary differential equation for the projected dynamics
\begin{equation}\nonumber
\dot x=v,\qquad \dot v=-\,\Xi(x,v),\qquad 
\Xi(x,v):=\dfrac{v^\top \nabla^2 U(x)\,v}{|\nabla U(x)|^2}\,\nabla U(x).
\end{equation}
For the Zig-Zag Sampler the same non-ergodic averaging leads to the piecewise-smooth rule developed in Sec.~\ref{sec:zzs_averaging}: coordinates with $\partial_j U(x)\ne0$ follow $v_j=-\mathrm{sgn}(\partial_j U(x))$, whereas coordinates with $\partial_j U(x)=0$ are chosen by a box-constrained quadratic program. The sign-flip inequality and boundary-layer bounds that justify this cycle averaging, together with a formal statement valid up to the first entrance time, are proved in Appendix~\ref{app:sec:proof_PDMP}.

\subsection{Computational cost for PDMP-based samplers}\label{subsec:comp-cost}

Throughout in this subsection, $\lambda_t$ denotes the intensity function and is therefore non-negative. In other sections, $\lambda_t$ denotes a signed pre-intensity (e.g. $\lambda_t=v_t^{\top}\nabla U(x_t))$, with the actual intensity taken as $\max\{0,\lambda_t\}$. 

For PDMP samplers the deterministic motion $\dot x=v$ is inexpensive; the dominant
cost comes from evaluating derivatives of $U$ when generating event times (gradients,
and Hessians). We simulate events on $[0,T]$ using thinning (\cite{Lewis1979SimulationThinning, Ogata1981}) with a
windowed envelope $\lambda_t^\ast \ge \lambda_t := \lambda(X_t,V_t)$; see also
\cite{Sutton02102023,Corbella2022,andral2024automated} for recent progress on PDMP implementations.
Let $M_T$ be the
number of candidate times generated by the envelope, $N_T$ the number of accepted
jumps, and $W_T$ the number of windows in the partition of $[0,T]$.

We adopt a simple window design: partition $[0,T]$ into windows $\Delta_k$ and use a
piecewise bound $\lambda_t^\ast$ on each window. Assume there exist window-uniform
constants $c_1,C_1,c_2>0$ such that for every window $\Delta_k$,
\[
c_1 \le \int_{\Delta_k}\lambda_t^\ast\,\dif t \le C_1
\qquad
c_2\int_{\Delta_k}\lambda_t^*\dif t\le \int_{\Delta_k}\lambda_t\dif t. 
\]
The former inequality is moderated by keeping the number of events in each window neither too small nor too large, while the latter inequality is governed by ensuring that the acceptance probability does not become too small.
By the compensator identity,
\[
\E[N_T]=\E\!\left[\int_0^T \lambda_t\,\dif t\right],
\qquad
\E[M_T]=\E\!\left[\int_0^T \lambda_t^\ast\,\dif t\right],
\]
hence
\[
c_2\,\E[M_T]\ \le\ \E[N_T]\ \le\ \E[M_T].
\]

To relate $M_T$ and $W_T$, note that summing the per-window bounds yields
\[
c_1\,W_T \ \le\ \int_0^T \lambda_t^\ast\,\dif t \ \le\ C_1\,W_T
\quad \text{almost surely},
\]
so
\[
c_1\,\E[W_T] \ \le\ \E[M_T] \ \le\ C_1\,\E[W_T].
\]
In implementation there is a fixed per-window overhead in addition to the per-candidate
derivative evaluations. Writing $C_{\mathrm{cand}}\ge 1$ for the number of gradient/Hessian
evaluations per candidate and $C_{\mathrm{win}}\ge 0$ for the number per window, the total
derivative count is
\[
\#\mathrm{deriv} \;=\; C_{\mathrm{cand}}\,M_T \;+\; C_{\mathrm{win}}\,W_T.
\]
Consequently,
\[
A\,\E[N_T] \ \le\ \E[\#\mathrm{deriv}] \ \le\ B\,\E[N_T]
\]
for  constants $A=C_{\mathrm{cand}}+C_{\mathrm{win}}C_1^{-1},B=C_{\mathrm{cand}}c_2^{-1}+C_{\mathrm{win}}c_1^{-1}c_2^{-1}>0$ that do not depend on $T$. Thus the expected derivative
cost is comparable, up to constants, to the expected number of jumps, and we use
$\E[N_T]$ as our cost proxy.

\begin{remark}
    In practice, the existence of the positive constant $c_2>0$ may fail locally when $\lambda_t= 0$. Instead,  we can expect 
    $$
    c_2\int_{\Delta_k}\lambda_t^*\dif t-\epsilon~|\Delta_k|\le \int_{\Delta_k}\lambda_t\dif t. 
    $$
    In this case, the expected number of derivatives is bounded above by
    $$
    B~(\E[N_T]+\epsilon T). 
    $$
    Hence if the tolerance $\epsilon$ is small enough, $\mathbb{E}[N_T]$ is dominant. 
\end{remark}

\section{Fluid limit of Bouncy Particle-family Samplers}\label{sec:fluid_bps}

This section analyses the fluid limit of the Bouncy Particle Sampler as a prototypical piecewise deterministic Monte Carlo algorithm. 
We first present the low-refreshment regime (Subsection~\ref{subsec:lowrefreshbps}), then derive the high-refreshment scaling (Subsection~\ref{subsec:highrefreshbps}), and finally obtain a constant-cost result for the Forward Event-Chain Sampler (Subsection~\ref{subsec:fec}).  Computational cost is measured throughout by the expected number of jump events, i.e., gradient evaluations.  
The Zig-Zag variants are treated separately in Section~\ref{sec:zzs_averaging}.

\subsection{Low-refreshment rate setting}
\label{subsec:lowrefreshbps}

We consider the Bouncy Particle Sampler for the potential $\varepsilon^{-1}U(x)$ and first analyse the case without velocity refreshment. 
The extended generator is $\mathcal{L}_{\text{BPS}}=\varepsilon^{-1}\mathcal{L}_{0,\text{BPS}}+\mathcal{L}_{1,\text{BPS}}$ where 
$$
\mathcal{L}_{0,\text{BPS}}f(x,v)=(v^{\top}\nabla  U(x))_+(\mathcal{B}-I)f(x,v),\quad \mathcal{L}_{1,\text{BPS}}f(x,v)=v^{\top}\nabla_xf(x,v)
$$
with reflection operator $\mathcal{B}$ defined by 
$$
\mathcal{B}f(x,v)=f(x,v-2n(x)n(x)^{\top}v).
$$
Observe that $\varepsilon^{-1}\lambda_+(x,v)$ is the intensity function of the Bouncy Particle Sampler, where 
$\lambda_+(x,v)=(v^\top\nabla U(x))_+$ and $\lambda(x,v)=v^{\top}\nabla U(x)$. 
For the analysis, we exclude the trivial case corresponding to states $(x,v)$ such that $\lambda(x,v)\neq 0$. In this case, the dynamics is simply $\dot{x}=v, \dot{v}=0$ possibly after an immediate jump until it hits $\lambda=0$. The non-trivial case of interest corresponds to the set of states where 
$$
\lambda(x,v)=0. 
$$
We refer to this set as the tangency set, as it represents the condition under which the system is about to undergo a qualitative change in dynamics. 
Near the tangency set, the fast part $\mathcal{L}_{0,\text{BPS}}$ flips $\lambda>0$ to $\lambda<0$; the free-flight governed by $\mathcal{L}_{1,\text{BPS}}$ makes $\dot{\lambda}>0$ under convexity,  so $\lambda$ drifts back to zero (see Subsection \ref{subsec:averaging-pdmp}). Consequently, the velocity continues to bounce frequently, resulting in a \textit{snapping} behaviour of the process.

We will say that $n(x)$ is a snapping direction at $x$ in the sense that the direction will be continuously switched in $\varepsilon^{-1/2}$ time and will be averaged in a long run. 
Relative to Section \ref{subsec:averaging-pdmp}, the decomposition is
\[
\renewcommand{\arraystretch}{1.05}
\begin{array}{c|c}
\textbf{slow variables} \quad& \quad \textbf{fast variable} \\ \hline
\bigl(x,\,v_1\bigr)& 
v_2
\end{array}
\]
where 
$$
v_1:=n^{\perp}(x)v,\quad v_2:=n(x)n(x)^{\top}v. 
$$
The norm of the velocity $v$ does not change throughout the dynamics. 
The Poisson equation $\mathcal{L}_{0,\text{BPS}}~u_0=0$ implies that $u_0$ is in the null space of $\mathcal{L}_{0,\text{BPS}}$.  In other words $u_0$ does not depend on the sign of $n(x)^{\top}v$. Moreover, the magnitude of this value, $|n(x)^{\top}v|$, which equals $|v_2|$, is determined by $|v_1|$ via the identity $|v_2|^2=|v|^2-|v_1|^2$.  Together these observations imply that $u_0$ is independent of $v_2$, that is, $u_0(x,v,t)=u_0(x, v_1,t)$. Then, the macroscopic equation $\dot{u}_0=\mathcal{L}_{1,\text{BPS}}~u_0$ becomes the effective equation 
\begin{align*}
\frac{\dif u_0}{\dif t}
&=v^{\top}\nabla_xu_0-\Xi(x, v)\nabla_{v_1}u_0, 
\end{align*}
where vector-valued function $\Xi:\mathbb{R}^N\times\mathbb{R}^N\rightarrow\mathbb{R}^N$ is defined by
\begin{equation}\nonumber
    \Xi(x,v)=\frac{v^{\top}\nabla^2U(x)v}{|\nabla U(x)|^2}~ \nabla U(x). 
\end{equation}
This equation corresponds to the flow $\phi_t(z)$ described as 
\begin{equation}\label{eq:bpsflow}
\dot{x}=v,\quad \dot{v}=-\Xi(x,v)
\end{equation}
where $z=(x,v)$ is an element of the tangency set.

Suppose that refreshment jumps occur with a constant rate $\rho>0$ at times
\[
0 = T_0 < T_1 < T_2 < \cdots,
\]
i.e., $T_{n+1}-T_n$ follows the exponential distribution with mean $\rho^{-1}$.  At each jump the velocity is resampled as $W_0, W_1, \dots$.  Let $\Phi_t^{\varepsilon}$ denote the corresponding flow of the slow variables 
$$
(X^{\varepsilon},n^\perp(X^\varepsilon)V^{\varepsilon}).
$$

We now construct the limiting flow $\Phi_t$ in $[T_n, T_{n+1})$ using the same jump times and velocity draws.
At $t = T_n$, set
\[
v(T_n) \;=\; B_+\bigl(x(T_n)\bigr)\,W_n,
\]
where the reflection operator $B_+(x)$ is defined by
\[
B_+(x)\,w = w \;-\; 2\,\bigl(w^\top n(x)\bigr)_+\,n(x),
\]
so that the post-jump velocity satisfies $v_{T_n}^\top n(x_{T_n}) \le 0$.
The motion evolves according to the free-flight
\begin{equation}\label{eq:freeflight}
\dot x_t = v_t, 
\quad
\dot v_t = 0,
\end{equation}
until the trajectory reaches the tangency set, at which point it follows the deterministic flow $\phi_t$.  
Figure \ref{fig:bps_2d} illustrates the limiting process. 
The proof of the following theorem is in  Appendix~\ref{app:subsec:bps_refresh_low}. 

\begin{theorem}\label{thm:bps_refresh_trajectory}
    Suppose that the both Bouncy Particle process and its limit share the same refresh-jump schedule. Suppose that $z^\varepsilon\rightarrow z$ for $z=(x,v)$ and $x\neq x^*$. Then, for $T>0$, the following convergence holds in probability: 
    $$
    \sup_{0\le t\le T}|\Phi_t^\varepsilon(z^\varepsilon)-\Phi_t(z)|\longrightarrow 0. 
    $$
\end{theorem}

\begin{figure}[ht]
  \centering
  \includegraphics[width=0.5\textwidth]{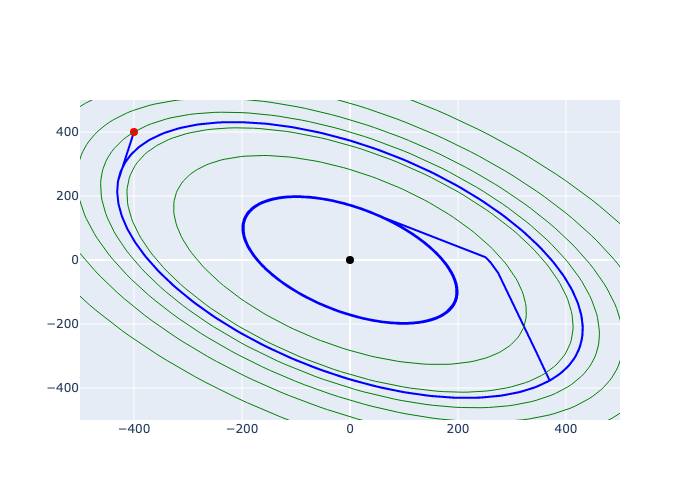}
  \caption{
Illustration of the Bouncy Particle Sampler. Aside from occasional refreshes, the limiting trajectory follows the level set between hits to the tangency set.”
}
  \label{fig:bps_2d}
\end{figure}

Below we present a complementary result showing that the Bouncy Particle Sampler requires only finitely many refresh jumps before it enters a high probability region. The proof can be found in Appendix~\ref{app:subsec:bps_no}.

\begin{theorem}\label{thm:bps_refresh}
For $\gamma>0$, 
suppose that $U(x)>U(x^*)+ \gamma$ and $x^\varepsilon \rightarrow x$. Then $n^\varepsilon(x^\varepsilon; \gamma)$ is tight. 
\end{theorem}

According to the theorem, we know that only a fixed number of refreshment jumps are needed before the process enters the high-probability region. To estimate how many total jumps are required, it is helpful to understand the number of bounce jumps that typically occur between refreshment steps. The following proposition provides such an estimate. Let $N_t^\varepsilon$ denote the number of jumps up to time $t$ before the process exits the band
$$
\bigl\{x \in \mathbb{R}^N : U(x^*) + 2\gamma < U(x) < U(x^*) + (2\gamma)^{-1}\bigr\}.
$$
This band condition is technically necessary because it provides uniform estimates for the process’s behaviour.
Its proof is deferred to  Appendix~\ref{app:subsec:bps_refresh_low}.

\begin{proposition}\label{prp:no_of_jumps}
Suppose that $\rho>0$ is either constant, or $\rho\rightarrow\infty$ and $\varepsilon^{1/2}\rho\rightarrow 0$. 
For $\gamma\in (0,1/2)$ and for $U(x^*)+2\gamma<U(x^\varepsilon
)<U(x^*)+(2\gamma)^{-1}$
$$
(1+\varepsilon^{1/2}\rho^{2})~\E[N_{T_1}^\varepsilon]
$$
is bounded above and below by positive constants as $\varepsilon\rightarrow 0$. 
\end{proposition}

\begin{remark}
The limiting dynamics describe a geodesic flow on the level set of the potential function, where $ v^\top \nabla U(x) = 0 $ is preserved. In other words, the process does not change the potential function without refreshment jumps. 
\end{remark} 


\begin{remark}[Computational cost] \label{remark:jump-count}
When $\rho>0$ is fixed, 
by Theorem \ref{thm:bps_refresh} and Proposition \ref{prp:no_of_jumps}, the computational cost (as discussed in Subsection~ \ref{subsec:comp-cost}) is estimated to be on the order of $\varepsilon^{-1/2}$.
\end{remark}

\begin{remark}[Comparison between Bouncy Particle Sampler and Random Walk Metropolis]
If the potential is subquadratic ($1<\beta<2$), the estimated cost is lower order for the Bouncy Particle Sampler, and if superquadratic ($2<\beta$) the cost of the random walk Metropolis is smaller (see Subsection \ref{subsec:setting}). 
\end{remark}

\begin{remark}\label{remark:non-convex-bps}
For target distributions with a non-convex potential function, snapping behaviour does not occur. Snapping arises since the process jumps at $\lambda=(v^{\top}\nabla U(x)) > 0$ and transitions to $-\lambda$ afterwards, while $\dot{\lambda}$ remains positive due to the convexity of the potential.
\end{remark}

\subsection{High-refreshment rate setting}
\label{subsec:highrefreshbps}

The low-refreshment setting is natural because, in practice, the scale factor $\varepsilon$ is unknown, so one cannot directly tune the refreshment frequency via $\varepsilon$.  On the other hand, as described in \cite{bierkens2018high}, the refreshment rate can be adjusted indirectly by controlling the relative number of bounce-induced and refreshment-induced jumps.  With this in mind, we study the regime in which the refreshment rate $\rho$ tends to infinity as $\varepsilon$ goes to zero, and we will identify the natural scaling of $\rho$ under that tuning strategy. For simplicity, we assume that $v$ is refreshed at time $t=0$. See Appendix~\ref{app:subsec:bps_high_ref} for the proof. 

\begin{theorem}\label{thm:bps_high_refresh}
Let $\gamma>0$ and let $U(x)> U(x^*)+ \gamma$. 
Suppose that $\rho\rightarrow \infty$ and $\varepsilon\rho\rightarrow 0$, and $x^\varepsilon\rightarrow x$ 
Then the $x$-process of the Bouncy Particle Sampler with the refreshment rate $\rho$ starting from $X^\varepsilon(0)=x^\varepsilon$ 
 converges in probability to 
 \begin{equation}\label{eq:bps_high_ref_flow}
 \dot{x}=-\frac{1}{\sqrt{2\pi}}n(x)
 \end{equation}
 uniformly on $0\le t\le \tau(x; \gamma)$. 
\end{theorem}

Overall, the expected jump count is the sum of refresh and bounce events.
Refreshments arrive at rate $\rho$ per unit time. Using Proposition~\ref{prp:no_of_jumps},
the bounce contribution per unit time scales as $O(\varepsilon^{-1/2}/\rho)$.
Thus the cost is
\[
f(\rho)\;\asymp\;\rho\;+\;\varepsilon^{-1/2}\rho^{-1}.
\]
Minimising $f(\rho)$ gives $\rho^\star=\varepsilon^{-1/4}$, i.e., balancing refreshes and bounces.

\begin{remark}
 By the above argument, the overall computational cost is estimated as $O(\varepsilon^{-1/4})$ when the refreshment rate is appropriately chosen. This can be achieved by adaptively tuning the refreshment rate so that the numbers of refreshments and bounces are approximately equal. Equalising the two contributions is the same strategy taken in  \cite{bierkens2018high}. 
\end{remark}

\subsection{The Forward Event-Chain Sampler}\label{subsec:fec}

As $\varepsilon\to0$ the Bouncy Particle Sampler suffers from
snapping: deterministic bounces trap the trajectory on one level set,
causing the jump count to explode. To avoid the high jump rate, the Forward Event-Chain type update  \cite{Michel2017} and \cite{vanetti2017piecewise} proposed the Forward Event-Chain update, which randomises the
normal component of the velocity and thereby breaks this snapping behaviour. 
This process replaces  the deterministic operator $\mathcal{B}$  of the Bouncy Particle Sampler  by the Markov transition kernel $\mathcal{Q}$ defined in (\ref{eq:fec}).

Here we briefly explain why the choice of $\mathcal{Q}$ preserves the invariant velocity distribution, specifically, the standard normal distribution.
Decompose the velocity as
$$
v= \xi~n(x)+w
$$
where $\xi\in\mathbb R$ is the gradient-aligned component and
$w := n^{\perp}(x)v$ is tangential.  Each Forward Event-Chain event consists of  
\begin{enumerate}
\item[(i)] redrawing $\xi$ from a Rayleigh$(1)$ distribution and $w$ from normal distribution on the tangent space, 
followed by  
\item[(ii)] flipping the sign of the new gradient-aligned component.  
\end{enumerate}
The tangential update of Step~(i), and Step~(ii) are obviously measure-preserving, so we analyse the gradient-aligned update in step~(i) only.

Since
$\lambda(x,v)=(v^{\top}\nabla U(x))_+ =
c\,\xi_+$ with $c:=|\nabla U(x)|$,
the jump generator acting on a test function $f$ of the gradient-aligned coordinate is
\[
  (\mathcal L f)(\xi)
  = c \int_{0}^{\infty} \xi_+\,(f(\xi^*)-f(\xi))\,R(\mathrm d\xi^*),
\]
where $R$ is the Rayleigh$(1)$ distribution.
Because
$\xi_+\,(2\pi)^{-1/2}\exp(-\xi^2/2)\,\mathrm d\xi
  =(2\pi)^{-1/2}\mathbf 1_{\{\xi\ge0\}}\,R(\mathrm d\xi)$,
we have
\[
  \int_{\mathbb R} (\mathcal L f)(\xi)\,\frac{1}{\sqrt{2\pi}}e^{-\xi^2/2}\,\mathrm d\xi = 0,
\]
so $\xi\sim\mathcal N(0,1)$ is invariant under step~(i).  
Consequently, the composite update (i) and (ii) leaves
normal distribution unchanged. For a complete proof see  Section A of \cite{Michel2017}, and Section A of \cite{vanetti2017piecewise}.

In the limit, the process reflects according to the above random update when it hits the tangency set, and the number of jumps is $O(1)$  
as illustrated in Figure \ref{fig:fec_2d}. 
More specifically, we have the following. 
See Appendix~\ref{app:subsec:proof_cs} for the proof. 

\begin{theorem}\label{thm:fec}
Let $\gamma>0$ and $U(x)> U(x^*)+ \gamma$. Suppose that $x^\varepsilon\rightarrow x$. Then $\{n^\varepsilon(x^\varepsilon; \gamma)\}_{\varepsilon>0}$ is tight. 
\end{theorem}

\begin{figure}[ht]
  \centering
  \includegraphics[width=0.5\textwidth]{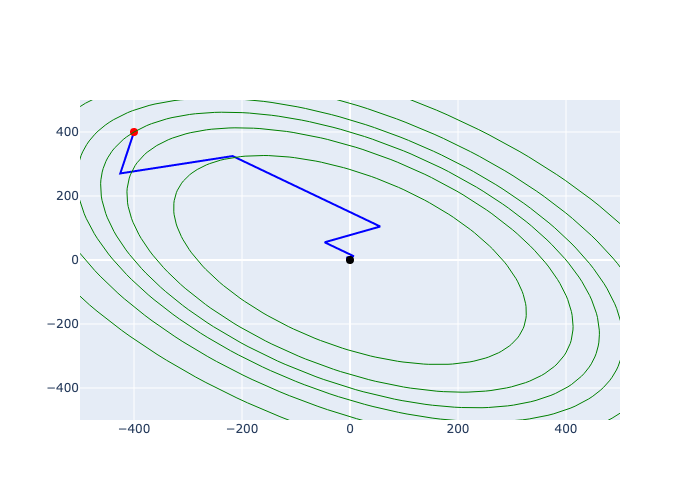}
  \caption{
Illustration of the Forward Event-Chain type update. The process does not exhibit a snapping phase. 
}
  \label{fig:fec_2d}
\end{figure}


\begin{remark}
Theorem~\ref{thm:fec} shows that the algorithm runs with $O(1)$ computational cost.  
Although the change from the original Bouncy Particle Sampler is only minor, it brings a substantial gain in efficiency in this context. 
\end{remark}

\section{Fluid limit of the Zig-Zag family Samplers}\label{sec:zzs_averaging}

This section establishes the fluid limit of the Zig-Zag Sampler. We begin by deriving the Karush--Kuhn--Tucker conditions, which expose the distinction between \textit{clipped} and \textit{snapping coordinates} in Section \ref{subsec:conv-opt-prob}.  Section~\ref{subsec:kkt_zigzag} shows how the optimiser fixes the velocity when the path reaches a tangency surface $\{\partial_n U=0\}$.  In Section \ref{sec:zigzag_diagonal_dominance}, we prove convergence of the Zig-Zag to a deterministic process  under a diagonal‑dominance assumption.

\subsection{Convex optimisation problem}
\label{subsec:conv-opt-prob}

For an integer $K>0$, 
let $H$ be a $K\times K$ symmetric positive-definite matrix, and let $c \in \mathbb{R}^K$ be a given vector. Consider the quadratic program
\[
\min_{v}  
\frac{1}{2}\,v^{\top} H\,v
+
c^{\top}v,
\]
subject to
\[
-1 \le v_n \le 1,
\quad
(n=1,\ldots,K).
\]
As we will see in Section~ \ref{subsec:kkt_zigzag}, this problem arises in the analysis of the local behaviour of the Zig-Zag Sampler, but for now we treat it in isolation.

To analyse the problem, we split the constraint $-1 \le v_n \le 1$ into two separate inequalities:
\[
v_n - 1 \le 0,
\quad
-\,v_n - 1 \le 0.
\]
Corresponding to these constraints, we define the Lagrangian
\[
L(v,\alpha,\beta)
=
\frac{1}{2}\,v^{\top} H\,v
+
c^{\top}v
+
\sum_{n=1}^K \alpha_n\,(v_n - 1)
+
\sum_{n=1}^K \beta_n\,(-\,v_n - 1),
\]
where $\alpha_n,\beta_n \ge 0\,$ for $n=1,\ldots,K$. Its solution $v^*$ satisfies the Karush--Kuhn--Tucker (KKT) conditions. First, $v^*$ satisfies the stationarity condition
\[
H\,v^*+c+\alpha-\beta=0.
\]
In addition, by complementary slackness, we have
\[
\alpha_n\,(v_n^* - 1) = 0,
\quad
\beta_n\,(v_n^* + 1) = 0,
\]
and therefore at least one of $\alpha_n^*$, $\beta_n^*$ should be equal to zero for the same index $n$. Hence, for each $n=1,\ldots,K$, exactly one of the following holds:
\begin{enumerate}
\item[(i)] $\alpha_n^* = -\,\bigl(H\,v^*\bigr)_n - c_n\ge 0$,  $\beta_n^* = 0$, and $v_n^* = +1$.
\item[(ii)] $\alpha_n^* = 0$, $\beta_n^* = \bigl(H\,v^*\bigr)_n + c_n\ge 0$, and $v_n^* = -1$.
\item[(iii)] $\alpha_n^* = \beta_n^* = 0$, and $-1 < v_n^* < 1$ and $\bigl(H\,v^*\bigr)_n + c_n = 0$.
\end{enumerate}

We call that any coordinate $n$ for which $v_n^* = \pm 1$ (cases 1 or 2 above) is a \emph{clipped coordinate}; write $\Kclipp$ for the set. The remaining indices, for which $-1 < v_n^* < 1$, are the \emph{snapping coordinates}; write $\Ksnapp$ for the set.  


\subsection{KKT condition and the limit process for the Zig-Zag Sampler}\label{subsec:kkt_zigzag}

We now present a heuristic derivation of how the Zig-Zag Sampler relates to the optimisation problem introduced above.
Suppose that the potential function $U(x)$ is twice continuously differentiable and $\nabla^2 U(x)$ is positive-definite. Recall that in the transient regime, the Zig-Zag Sampler switches the $n$th coordinate at a rate $ \varepsilon^{-1} (v_n\partial_n U(x))_+$.
In the small $\varepsilon$ limit, this means that the Zig-Zag Sampler dynamics for the $n$-th coordinate are described by
\[
  \dot{x}_n \;=\; -\operatorname{sgn}\bigl(\partial_n U(x)\bigr)
  \quad\text{whenever}\quad \partial_n U(x)\neq 0,
\]
since the velocity $v_n$ will `immediately' switch to this direction.
Accordingly, we are interested in the set of \emph{tangency coordinates} defined by
\[
  K=K(x) \;=\; \{\,n : \partial_n U(x)= 0\}.
\]
Set $J=K^c$. 
For $n \in J$, the dynamics is given as above, so $v_n = -\operatorname{sgn}(\partial_n U(x))$ for $n \in J$.
Suppose now that  $n\in K$ at  time $t= 0$. One has
\[
  v_n(t)\,\partial_n U\bigl(x(t)\bigr)
  \;\approx\;
  t\,v_n(0)\,\bigl[\nabla^2 U\bigl(x(0)\bigr)\,v(0)\bigr]_n+O(t^2)
  \quad\text{for small } t.
\]
Hence, if $v_n\bigl(\nabla^2 U(x)\,v\bigr)_n > 0$, the velocity component $v_n$ flips sign almost `immediately'.  Consequently, the velocity $v_n$,  viewed here as the \emph{averaged} velocity of the $n$-th coordinate of $x$, satisfies
\begin{equation}\label{eq:kkt_negativity}
  v_n \,\bigl(\nabla^2 U(x)\,v\bigr)_n \;\le\; 0
\end{equation}
if $n$ stays in the tangency coordinates. 
By the dynamics of the Zig-Zag Sampler, this implies that 
\begin{equation}\nonumber
  \bigl(\nabla^2 U(x)\,v\bigr)_n \neq 0\quad \Longrightarrow\quad v_n \;=\; -\,\operatorname{sgn}\bigl[\bigl(\nabla^2 U(x)\,v\bigr)_n\bigr].
\end{equation}
On the other hand, when $(\nabla^2 U(x)\,v)_n=0$, the averaged $n$-th component $v_n$ need not lie in $\{-1,+1\}$; instead, it may lie anywhere in $[-1,+1]$.

We will now establish that the averaged velocity  in the tangency coordinates, $v_K = (v_k)_{k \in K}$, can be interpreted as a solution to a quadratic program. 
Define
\begin{align*}
  H & \;=\; H(x) \;=\; (H_{kl})_{k,l \in K} \;=\; \bigl(\partial_k\partial_l U(x)\bigr)_{k,l\in K}, \\
  c & \;=\;  c(x) \;=\; (c_{k})_{k \in K} \;=\; \Bigl(\sum_{l\in J}\partial_k\partial_l U(x)\,v_l\Bigr)_{k\in K} \;=\; \Bigl(-\sum_{l\in J}\partial_k\partial_l U(x)\,\operatorname{sgn}\left(\partial_l U(x)\right)\Bigr)_{k\in K}.
\end{align*}
For $k\in K$, let $\alpha_k=\left[-(Hv+c)_k\right]_+$ and 
$\beta_k=\left[(Hv+c)_k\right]_+$. 
For each $k\in K$, exactly one of the following holds: 
\begin{enumerate}
    \item[(i)] $\alpha_k > 0$ and $\beta_k = 0$, then then $(Hv+c)_k = (\nabla^2U(x) v)_k < 0$ and therefore $v_k = +1$.
    \item[(ii)] $\alpha_k = 0$ and $\beta_k > 0$, then $(Hv+c)_k = (\nabla^2U(x) v)_k > 0$ and therefore $v_k = - 1$.
    \item[(iii)]
    $\alpha_k = \beta_k = 0$ then $(Hv+c)_k =(\nabla^2 U(x) v)_k = 0$ and  the average velocity satisfies $-1 \le v_k \le 1$.
\end{enumerate}
This corresponds exactly to the KKT conditions discussed in the previous section. We conclude that 
$v_K$ is the solution to the quadratic program in Subsection \ref{subsec:conv-opt-prob}. 

From the formal argument above, the limiting dynamics of the Zig-Zag Sampler can be expressed as
$$\dot x = v$$ with 
\begin{align*}
  v_j &  \;=\; -\operatorname{sgn}\bigl(\partial_j U(x)\bigr),
  \quad
  j\in J, \\
  v_K & \;=\; \arg\min  \left\{ \frac{1}{2}v_K^{\top}Hv_K+c^{\top}v_K:  -1\le v_k\le 1,\quad k\in K \right\}.
\end{align*}
A schematic, piecewise-smooth realisation is summarised in Procedure~\ref{proc:update}.

\begin{figure}[htbp]
\centering
\begin{minipage}{0.90\linewidth}\small        
\textbf{Procedure \refstepcounter{algorithm}\thealgorithm\label{proc:update} 
(Coordinate-update rule)}

\begin{enumerate}\setlength{\itemsep}{4pt}
  \item[(i)] \textbf{Initial split.}\\
        $K \gets \{n : \partial_n U(x)=0\}$ \;(tangency),\quad
        $J \gets \{1,\dots,N\}\!\setminus\! K$ \;(non-tangency).
  \item[(ii)] \textbf{Assign velocities.}
        \begin{enumerate}\setlength{\itemsep}{2pt}
          \item[(a)] If $n\in J$: $\,v_n \gets -\operatorname{sgn}\!\bigl(\partial_n U(x)\bigr)$.
          \item[(b)] If $n\in K$: \textit{solve} the convex program  
                    (Sec.~\ref{subsec:conv-opt-prob}) to obtain $v_K$.
        \end{enumerate}
        \textbf{Evolve:} integrate $\dot x=v$ until next event. 
  \item[(iii)] \textbf{Boundary event.}\\
        When some $k\in J$ first satisfies $\partial_k U(x)=0$, set  
        $K \gets K\cup\{k\},\quad J \gets J\setminus\{k\}$  
        and \emph{return to Step 2}.
\end{enumerate}
\end{minipage}
\end{figure}

We classify tangency coordinates as either clipped or snapping, following the terminology introduced for the optimisation problem. We classify each coordinate according to Table \ref{tab:tangency-classes}. 

\begin{table}[htbp]
\centering
\renewcommand{\arraystretch}{1.05}

\begin{tabular}{lccc}
\toprule
& \textbf{Non-tangency ($J$)} & \multicolumn{2}{c}{\textbf{Tangency ($K$)}}\\
&                       & Clipped ($\Kclipp$) & Snapping ($\Ksnapp$)\\
\midrule
$\partial_n U(x)$
& $\neq 0$               & $=0$   & $=0$\\
Average velocity $v_n$
& $-\operatorname{sgn}\!\bigl(\partial_n U(x)\bigr)$ 
& $\;\pm1$             & $(-1,1)$\\
\bottomrule \\
\end{tabular}
\caption{Classification of coordinates in the Zig-Zag fluid limit. 
If the gradient component is non-zero (left column) the sampler moves strictly against it.  
When $\partial_n U(x)=0$ (tangency) the optimiser either drives the velocity to a boundary value ($\pm1$, \textbf{clipped}) or leaves it in the interior $(-1,1)$ (\textbf{snapping}).}
\label{tab:tangency-classes}
\end{table}

\subsection{Diagonal dominant scenario}\label{sec:zigzag_diagonal_dominance}

We consider the following diagonal-dominance condition: 
\begin{equation}
    \label{DD}
\partial_i^2U(x)-\sum_{j\neq i, j=1,\ldots, N}|\partial_i\partial_j U(x)|>0
\end{equation}
for any tangency coordinate $i$. 
Under this condition, there is no clipped coordinate, and set of the tangency set coincides with the snapping set. Thus the set $\{1,\ldots, N\}$ partitioned into the snapping coordinates denoted by $K$ since $K=\Ksnapp$ in this case, and the non-tangency coordinate denoted by $J$.  

If we apply the averaging argument of Section \ref{subsec:averaging-pdmp},
the state splits into
\[
\renewcommand{\arraystretch}{1.05}
\begin{array}{c|c}
\textbf{slow variables}\quad  &  \quad\textbf{fast variable} \\ \hline
(x,\,v_{J}) & v_{K}
\end{array}.
\]
The extended generator $\mathcal{L}_{\text{ZZS}}^\varepsilon$ has a decomposition $\varepsilon^{-1}\mathcal{L}_{0,\text{ZZS}}+\mathcal{L}_{1,\text{ZZS}}$ such that
$$
\mathcal{L}_{0,\text{ZZS}}f(x,v)=\sum_{i=1}^N(v_i\partial_i U(x))_+(\mathcal{F}_i-I)f(x,v),\quad 
\mathcal{L}_{1,\text{ZZS}}f(x,v)=v^{\top}\partial_xf(x,v). 
$$
As in the Bouncy Particle Sampler, in a neighbourhood of the tangency set
$\bigcup_{k\in K}\{\partial_k U(x)=0\}$,
the fast dynamics flip any positive
$\lambda_k:=v_k\,\partial_k U(x)$ to negative, while the free-flight dynamics
generated by $\mathcal{L}_{1,\mathrm{ZZS}}$ imply $\dot{\lambda}_k>0$ under the
diagonal-dominance condition \eqref{DD}. This closes a micro-cycle (see
Subsection~\ref{subsec:averaging-pdmp}).

For $j\in J$, the fast dynamics enforce
$v_j=-\operatorname{sgn}(\partial_j U(x))$. For $k\in K$, averaging over the
micro-cycle keeps the tangency constraints invariant along the slow motion:
\[
\frac{\mathrm d}{\mathrm dt}\,\partial_k U(x(t))
=\sum_{j=1}^N \partial_k\partial_j U(x)\,v_j = 0,\qquad k\in K.
\]
Since $v_J=(v_j)_{j\in J}$ is fixed, the admissible slow direction for
$v_K=(v_k)_{k\in K}$ is determined by
$H_{KK}(x)\,v_K = -\,H_{KJ}(x)\,v_J$, hence
\[
v_K \;=\; -\,H_{KK}(x)^{-1} H_{KJ}(x)\,v_J,
\]
where
\[
H(x):=\nabla^2 U(x)=
\begin{pmatrix}
H_{KK}(x)&H_{KJ}(x)\\
H_{JK}(x)&H_{JJ}(x)
\end{pmatrix}.
\]



Until the process reaches the set $\{x:U(x)\le U(x^*)+\gamma \}$, it evolves according to  Procedure~\ref{proc:update}. In this case, the optimisation problem is solved explicitly. 

Let $\phi_t(x)$ denote the deterministic flow defined by the rule. By construction, the hitting time to the high-probability region, $\tau(x; \gamma)$, is finite, as each coordinate $n$ advances at unit speed and never flips sign until it reaches the boundary $\partial_n U(x) = 0$, beyond which it remains confined to that constraint surface.
The proof of the following theorem is in  Appendix~\ref{app:subsec:zigzag_trajectory}.

\begin{theorem}\label{thm:zigzag_trajectory}
Assume \ref{assumption} and $U(x)$ is diagonally dominant. 
Let $\gamma>0$. 
Let $x$ be such that $U(x)>U(x^*)+ \gamma$ and let $x^\varepsilon\rightarrow x$. The Zig-Zag Sampler starting from $x^\varepsilon$ and $v\sim\mathcal{U}(\{-1,+1\}^N)$ converges in probability to $\phi_t(x)$ uniformly for $t\le \tau(x; \gamma)$. 
\end{theorem}

Let $N_{i,t}^\varepsilon$ be the number of jumps corresponding to the $i$-th coordinate. The following proposition is a part of  Corollary~\ref{cor:zz_constraint}. 

\begin{proposition}
Suppose that $\nabla_i U(x^\varepsilon)=0$, and $U(x^*)+\gamma<U(x^\varepsilon)<U(x^*)+\gamma^{-1}$ for some $\gamma\in (0,1)$. 
For $T>0$, 
$$
\varepsilon^{1/2}~\E[N_{i, T\wedge \tau^\varepsilon(x^\varepsilon;\gamma)}^\varepsilon]
$$
is bounded away from zero and infinity as $\varepsilon\rightarrow 0$. 
\end{proposition}

\begin{remark}
The number of jumps is of the order $ \varepsilon^{-1/2} $, as in the Bouncy Particle Sampler. A key feature of the Zig-Zag Sampler is that it does not require a tuning parameter to reach near the centre, unlike the random walk Metropolis algorithm and the Bouncy Particle Sampler.
\end{remark}

\begin{remark}\label{remark:non-convex-zz}
As in the case of the Bouncy Particle Sampler, the snapping behaviour does not occur when the potential function is not convex. Snapping arises because the sum of (\ref{eq:kkt_negativity}), i.e., $\sum v_n(\nabla^2 U(x)v)_n$ is positive. If the sum is positive, the total sum of $\lambda_n = v_n \, \partial U(x)$ for $n=1,\ldots, N$ is increasing over time, where $\lambda_n^+$ is the intensity function of the $n$-th coordinate. Whenever $\lambda_n$ becomes positive, it immediately flips to $-\lambda_n$. However, because the overall sum remains positive and keeps increasing, another coordinate $m$ will eventually attain a positive $\lambda_m$ and flip as well.
\end{remark}

\subsection{Non-diagonal-dominant scenario (illustrative)}

We demonstrate how the system's dynamics change in  Procedure~\ref{proc:update} for non-diagonally-dominant Gaussian case. Typically, the process adds the new tangency coordinate as a  snapping coordinate each time it intersects a new tangency surface $\{\partial_n U=0\}$ as in the diagonal-dominance scenario. So the number of tangency coordinates is non-decreasing. However, this is not always the case. 

Consider the two-dimensional case with the Hessian:  
\[
\nabla^2U(x)\equiv H=\begin{pmatrix}
    0.4& 0.5\\
    0.5 &1
\end{pmatrix}. 
\]
When $K=\emptyset$, the process follows the direction $v\in\{-1,+1\}^2$ until it hits a hyperplane $(Hx)_1=0$ or $(Hx)_2=0$. Suppose that it hits $(Hx)_1=0$ at time $t$, which implies that $t\mapsto v_1(H(x+vt))_1$ is increasing, i.e., $v_1(Hv)_1>0$, so $v_1=v_2$. At this point, the tangency coordinate is $K=\{1\}$ at time $t$, and the corresponding optimisation problem is 
$$
\min_{-1\le v_1\le 1}\frac{1}{2}H_{11}v_1^2 +c~ v_1,\quad c=H_{12}v_2.
$$
The solution to this problem is $v_1=-v_2$ in this case. In this case, the coordinate $1$ is clipped, and the process
reverses direction rather than becoming more constrained. 
The tangency coordinate is $\emptyset$ after the hitting time.  However, when the process next hits to the hyperplane $(Hx)_2=0$, it becomes trapped in the hyperplane, i.e., $K=\{2\}$ and $2$ is a snapping coordinate. 
See
Figure \ref{fig:zz_2d_diagonal_non_dominant}. 

\begin{figure}[ht]
  \centering
  \includegraphics[width=0.5\textwidth]{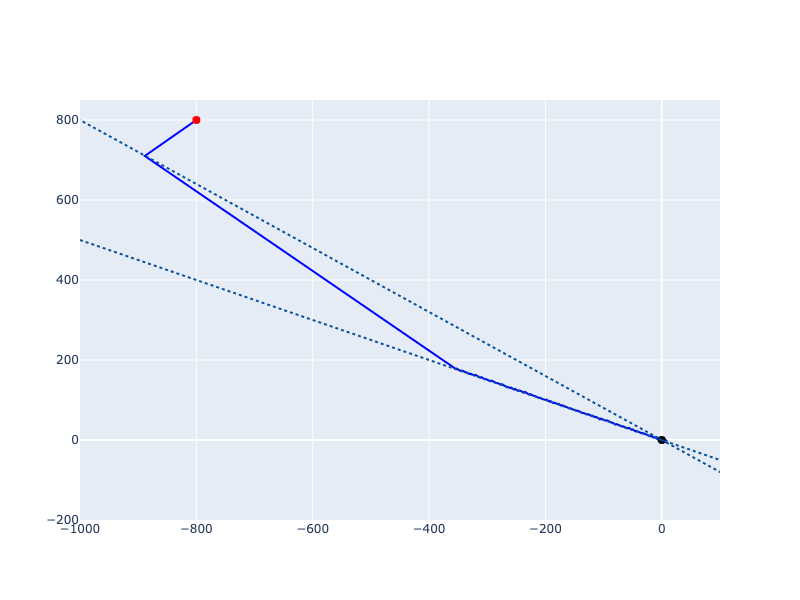}
  \caption{
Illustration of the Zig-Zag Sampler in the ``change direction'' clipped case. The starting position is indicated by the red dot.}
  \label{fig:zz_2d_diagonal_non_dominant}
\end{figure}

In the setting above, when the trajectory meets the first hyperplane, the number of tangency coordinates remains unchanged immediately before and after the hitting time. In a typical setup, this count usually increases, just as it does at the second hyperplane. There are, however, rare cases in which the count actually decreases after the intersection, compared with its value just before the hitting time. Consider the following Hessian: 
\[
\nabla^2 U(x)\equiv H=\begin{pmatrix}
1 & -0.3 & 1 \\
-0.3 & 1 & 1 \\
1 & 1 & 4
\end{pmatrix}. 
\]
Suppose that $v_2=v_3=+1$ and $K=\{1\}$, meaning that $(Hx)_1=0$ is satisfied. 
The optimal solution for $v_1$ is $v_1=-0.7$, that is, $1$ is a snapping coordinate. 
When the process hits $(Hx)_2=0$ at time $t$, the optimisation problem is:  
$$
\min\frac{1}{2}\tilde{v}^{\top}\tilde{H}\tilde{v}+
\tilde{c}^{\top}
\tilde{v},\quad \tilde{v}=(v_1,v_2)\in [-1,+1]^2
$$
where 
$$
\tilde{H}=\begin{pmatrix}
1&-0.3\\
-0.3&1
\end{pmatrix}
,\quad 
\tilde{c}=\begin{pmatrix}
1\\ 1
\end{pmatrix}. 
$$
The optimal solution for this problem is $\tilde{v}=(-1, -1)$, which results in the set of tangency coordinates becoming $\emptyset$. It means that, after hitting $(Hx)_1=0$, the process follows $(Hx)_1=0$, but when it next hits $(Hx)_2=0$, it escapes from both hyperplanes, as illustrated in Figure \ref{fig:zz_3d_escape}. 

\begin{figure}[ht]
  \centering
  \includegraphics[width=0.8\textwidth]{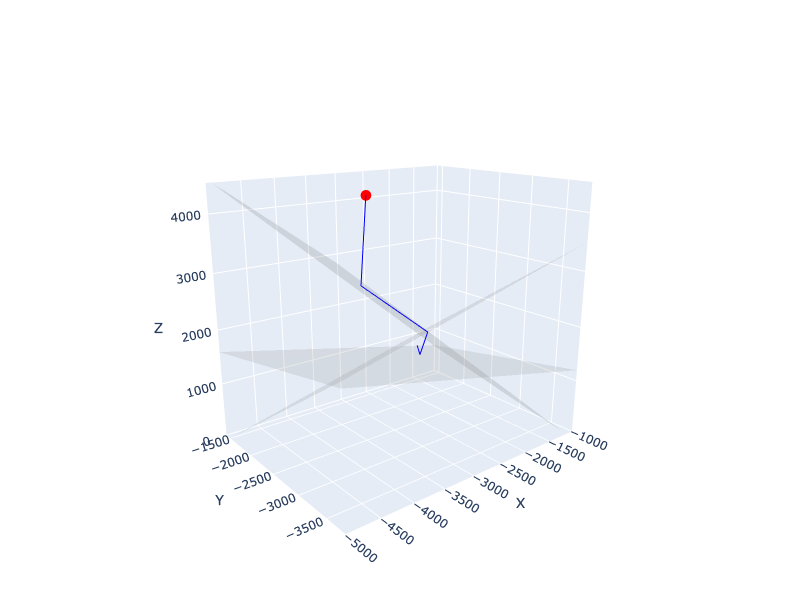}
  \caption{
Illustration of the Zig-Zag Sampler in the ``escape'' clipped case. First, the process hits $(Hx)_1=0$ and the process becomes trapped in the hyperplane. When the process hits $(Hx)_2=0$ next, the line in blue escapes from both hyperplanes, which are shown in grey. Then, it hits $(Hx)_3=0$ and becomes trapped in the third hyperplane. 
}
  \label{fig:zz_3d_escape}
\end{figure}

\subsection{Coordinate Sampler}

To avoid the $\mathcal O(\varepsilon^{-1/2})$ cost we turn to the Coordinate Sampler, which flips a single velocity at each event. In contrast to the Zig-Zag Sampler, the Coordinate Sampler updates only one velocity at each jump, so the intensity function does not snap, even under a convex potential. See Figure \ref{fig:zz_cs_2d} for the corresponding output.

Formally, the Coordinate Sampler changes only one coordinate $n$ at a time, and 
on this coordinate, the process moves in direction $s\in\{\pm 1\}$. The rate function is $(s\partial_nU(x))_+$. 
When an event occurs, a new coordinate $m$ is chosen proportional to $|\partial_mU(x)|$, 
and its sign is updated by $s=-\operatorname{sgn}\partial_mU(x)$. The limit process is more or less the same, but the event occurs when the process hits $\partial_nU(x)=0$ when the coordinate $n$ is selected. 
Let $n^\varepsilon(x; \gamma)$ be the number of events until the process hits $\{x: U(x)\le U(x^*)+\gamma \}$. The proof of the following Theorem is in Appendix~\ref{app:subsec:proof_cs}. 

\begin{theorem}\label{thm:cs}
    Let $\gamma>0$ and $U(x)> U(x^*)+ \gamma$. Suppose that $x^\varepsilon\rightarrow x$. Then $\{n^\varepsilon(x^\varepsilon; \gamma)\}_{\varepsilon>0}$ is tight. 
\end{theorem}


\begin{figure}[ht]
  \centering
  \includegraphics[width=0.4\textwidth]{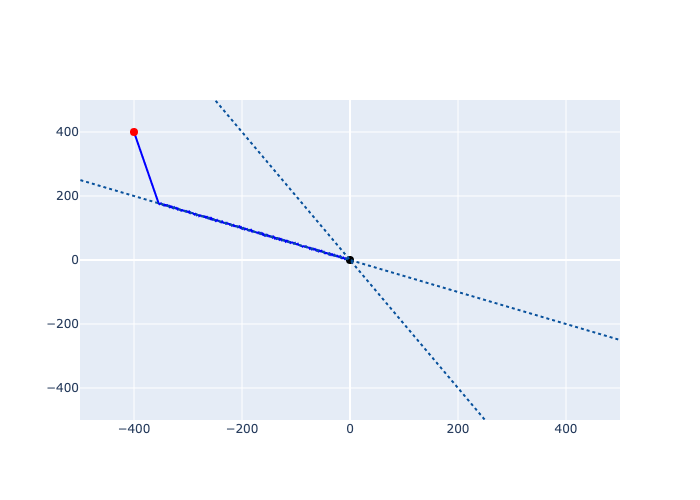}
  \includegraphics[width=0.4\textwidth]{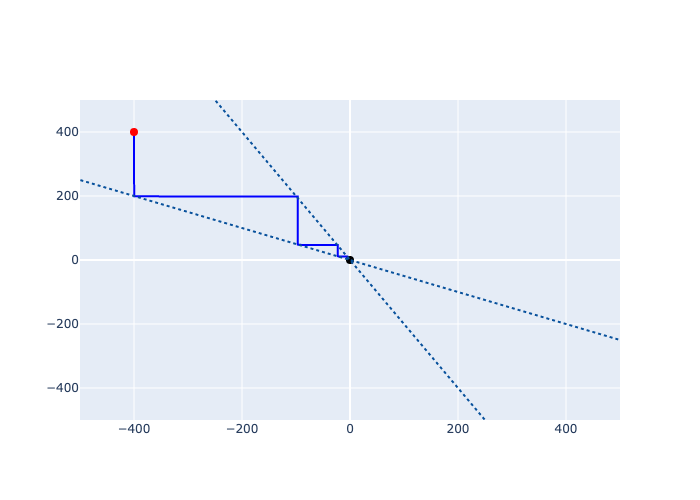}
  \caption{
Illustration of the Zig-Zag Sampler (left) and the Coordinate Sampler (right) is shown. The dotted lines indicate the boundaries $(Hx)_1 = 0$ and $(Hx)_2 = 0$. The Zig-Zag Sampler becomes trapped along one of these boundaries, whereas the Coordinate Sampler simply reflects off the boundaries without getting stuck.
}
  \label{fig:zz_cs_2d}
\end{figure}

\begin{remark}
    Theorem \ref{thm:cs} shows that the number of jumps is $O(1)$, as in the Forward Event-Chain Sampler. 
\end{remark}

\section{Discussion}\label{sec:discussion}

Our investigation into piecewise deterministic Markov processes (PDMPs)  under convex potentials provides valuable insight into their use as Monte Carlo methods. By analysing the transient phase of the Bouncy Particle Sampler and the Zig-Zag Sampler within a fluid-limit framework enables a direct efficiency comparison to traditional Metropolis--Hastings algorithms.

A key finding is that for Gaussian target distributions, the event complexity of the BPS with a low refreshment rate, as well as that of the Zig-Zag Sampler, is comparable to that of classical random-walk Metropolis algorithms, with complexity $O(\varepsilon^{-1/2})$. This reduced efficiency is due to snapping, caused by snapping, i.e., rapid reflections that pin the path to a constraint surface. However, the efficiency of BPS can be improved to $O(\varepsilon^{-1/4})$ by tuning the refreshment rate.
In addition, we show that both the Forward Event-Chain sampler \cite{Michel2017} and the Coordinate Sampler \cite{Wu2020CoordinateSampler} can achieve $O(1)$ computational complexity, which highlights the performance benefits possible with carefully designed PDMP-based samplers.
As in Remark \ref{remark:jump-count}, these results assume that the samplers are properly implemented (see \cite{Sutton02102023, Corbella2022, andral2024automated}).

Our results rely on convexity of the potential functions. 
Notably under non-convex potentials, BPS and Zig-Zag can avoid snapping. Non-convexity presents challenges for standard random-walk Metropolis or gradient-based methods, yet PDMP-based samplers appear capable of circumventing some of these difficulties. Quantifying performance along the lines presented in this work for non-convex landscapes remains a direction for further research.

In this study, we employed a simplified setting for the Zig-Zag Sampler, which does not yet exploit the full optimisation‐based description. A more general formulation could yield deeper theoretical insight into its behaviour, particularly in anisotropic settings.

More broadly, our work illustrates the promise of fluid limit analysis for the study of PDMP-based samplers.  Our results show how the transient dynamics of such samplers can be effectively captured using suitable averaging arguments. This perspective opens new directions for algorithm design; for instance, the Forward Event-Chain and the Coordinate Sampler serve as concrete examples where fluid limit insights guide the development and selection of efficient sampling strategies. 

In conclusion, this work provides a rigorous analytical framework for understanding the transient regime of PDMP-based samplers and offers practical guidance for designing efficient samplers. Moreover, the same fluid-limit lens clarifies behaviours of traditional MCMC schemes, enriching the broader theory of Monte Carlo efficiency.

\section*{Acknowledgements}
SA was supported by the EPSRC studentship grant EP/W523793/1. JB was funded by the research programme ‘Zigzagging through computational barriers’ with project number 016.Vidi.189.043, which was financed by the Dutch Research Council (NWO). KK was supported in part by JST, CREST Grant Number JPMJCR2115, Japan, and JSPS KAKENHI Grant Number  JP21K18589.
GOR was supported by EPSRC grants Bayes for Health (R018561) CoSInES (R034710) PINCODE (EP/X028119/1), EP/V009478/1 and also
UKRI for grant EP/Y014650/1, as part of the ERC Synergy project OCEAN.

\appendix

\section{Proof of the averaging principle of PDMP-based samplers}
\label{app:sec:proof_PDMP}

Our objective is to show that the perturbed PDMP trajectory $\bigl(X_t^{\varepsilon}\bigr)_{t\ge0}$ remains uniformly close to the fluid limit trajectory up to the first time it enters the high-probability set $L_\gamma$: 
$$
L_\gamma = \{x\in\mathbb{R}^N:  U(x)\le U(x^*)+\gamma\}. 
$$
The proof proceeds in three steps:
\begin{enumerate}
    \item[(i)] We confine the analysis to the ``energy band'' $B_\gamma$:
    $$
B_\gamma =\{x\in\mathbb{R}^N: U(x^*)+\gamma\le U(x)\le U(x^*)+\gamma^{-1}\}. 
$$
    \item[(ii)] Uniform gradient and Hessian bounds on $B_{\gamma}$ yield Lipschitz‐type growth control on the drift and the rate functions. 
    \item[(iii)] The sign-flip inequality (Lemma \ref{lem:monotone}) converts those controls into an $L^{1}/L^{\infty}$ estimate on the gap between the perturbed trajectory and its limit. 
\end{enumerate} 
Together these steps, applying Gronwall’s inequality finishes the proof of  Theorem~\ref{thm:bps_refresh_trajectory} and Theorem~\ref{thm:zigzag_trajectory}. Similarly, Martingale convergence framework completes the proof of  Theorem~\ref{thm:bps_high_refresh}. The uniform bounds on the band are also used to estimate the number of jumps in Proposition \ref{prp:no_of_jumps}.

Define the first exit time of the perturbed process $X_t^\varepsilon$ from the band by
$$
\tau^\varepsilon(x; \gamma)=\inf\{t\ge 0: X^\varepsilon_t\notin B_{\gamma} \}
$$
and write simply $\tau^\varepsilon$ when $(x,\gamma)$ is clear.

Under Assumption \ref{assumption}, there exist positive constants
\begin{equation}\label{eq:local_uniform_bound}
\underline{\eta}\le |\nabla U(x)|\le \overline{\eta},\quad 
\underline{\kappa}~I\le \nabla^2 U(x)\le \overline{\kappa}~I
\end{equation}
for $x\in B_{\gamma} $, where $I$ is the identity matrix with size $N$. Also, the upper bound of the norm of a function $f$ on $B_{\gamma} $ is denoted by $C_f$. 
The Hessian lower bound is necessary as it controls the hitting time of the process to tangency set. The lower and upper bound controls the behaviour of the limit process.


\begin{lemma}\label{lem:monotone}
     Let $f: [0,T) \to \mathbb{R}$ be right-continuous  with  a finite set of discontinuities $D(f)$ and satisfying $f(0)=0$.  
     At each jump time $t\in D(f)$ we have $f(t-)>0$ and 
 $f(t-)=-f(t)$. On the complement $D(f)^{\mathrm{c}}$ the function is differentiable and  obeys \[
m\le f'(t)\le M,\quad 0<m\le M. 
\]
Then 
\begin{enumerate}
    \item[(i)] 
$\sup_{0\le t\le T}|f(t)|^2\le 2M~\int_0^T(f(t))_+\dif t$.
\item[(ii)] 
Let $N_T=\sharp D(f)$. Then 
\[
\frac{m^2}{2M}\left(1+\frac{M}{m}\right)^{-1}\frac{T^2}{1+N_T}~\le~ \int_0^T |f(t)|\dif t~\le~ \left(1+\frac{M}{m}\right)\int_0^T (f(t))_+\dif t. 
\] 
\end{enumerate}
\end{lemma}

\begin{proof}
Without loss of generality, assume that $f$ vanishes at both endpoints; this costs at most one additional discontinuity. To achieve this, we extend the domain of $f$ appropriately. If $f(T-) = 0$, we extend the domain to $[0, T]$ and define $f(T) = 0$. If $f(T-) \neq 0$, we define $f(T) = -|f(T-)|$, and extend the domain to $[0, T + |f(T-)|/m]$, setting $f(t) := f(T) + m(t - T)$ for the extended region.

\begin{enumerate}
    \item[(i)] 
For every $t\in D(f)$, we know that $f(t-)>0$, and we can find two points $s$  and $u$
in the interval such that $s<t<u$ with $f(s)=f(u)=0$ and $f$ is continuous on $[s,t)$ and on $(t,u]$. Because $f'$ is bounded on intervals of continuity,
\begin{equation}\label{eq:interval-bound}
f(t-)=\int_s^tf'(v)\dif v=\int_t^u f'(v)\dif v~\Longrightarrow~
 t-s, 
u-t\in\left[\frac{f(t-)}{M}, \frac{f(t-)}{m}\right]. 
\end{equation}
From this, we deduce the bound of the integral as follows, since each area is bounded by a triangle with height $f(t-)$ with base lengths $t-s$ and $u-t$, respectively: 
$$
\frac{f(t-)^2}{2M}\le \int_s^t f(x)\dif x,\quad
\int_t^u (-f(x))\dif x\le \frac{f(t-)^2}{2m}. 
$$
Summing over all points $t\in D(f)$, we obtain the following inequalities: 
$$
\sum_{t\in D(f)}\frac{f(t-)^2}{2M}\le \int_0^T(f(t))_+\dif t,\quad 
\int_0^T(-f(t))_+\dif t\le 
\sum_{t\in D(f)}\frac{f(t-)^2}{2m}. 
$$
This proves the first claim and gives the upper bound for the integrals of $|f(x)|$, since $\sup |f(t)|$ is attained at $t\in D(f)$. 

\item[(ii)] 
Now using (\ref{eq:interval-bound}), we can relate the length of the intervals: 
$$
(u-s)\le \left(1+\frac{M}{m}\right)(u-t). 
$$
This gives the lower bound of $f(t-)$ using the interval $(u-s)$ in place of $(t-s)$ and $(u-t)$: 
$$
m\left(1+\frac{M}{m}\right)^{-1}(u-s)\le f(t-)
$$
Next, we introduce the dependence of $s$ and 
$u$ on $t$, denoting them as $s(t)$ and $u(t)$. With this, we can now express the total time $T$
as the sum since the interval $[s(t), u(t)]$ covers $[0,T]$: 
$$
T\le \sum_{t\in D(f)}(u(t)-s(t)). 
$$
Additionally, we can estimate the square of the differences by using the Cauchy--Schwarz inequality: 
$$
\sum_{t\in D(f)}(u(t)-s(t))^2\ge 
\left(\sum_{t\in D(f)}(u(t)-s(t))\right)^2~\left(~\sum_{t\in D(f)}1\right)^{-1}\ge \frac{T^2}{N_T}. 
$$
Finally, combining all of the above estimates, we get the key result: 
\begin{align*}
    \sum_{t\in D(f)}\frac{f(t-)^2}{2M}&\ge 
    \frac{m^2}{2M}\left(1+\frac{M}{m}\right)^{-1}\sum_{t\in D(f)}(u(t)-s(t))^2\\
    &\ge \frac{m^2}{2M}\left(1+\frac{M}{m}\right)^{-1}\frac{T^2}{N_T}. 
\end{align*}
Thus, we have established the desired bound since $\int_0^T|f(t)|\dif t\ge \int_0^T(f(t))_+\dif t$, observing that the function $f$ on the extended interval has at most one additional discontinuity.
\end{enumerate}
\end{proof}

\subsection{Proof of Proposition \ref{prp:no_of_jumps}}\label{app:subsec:bps_no}

Throughout the proof let
$$
\lambda_t(x,v):= v^{\top}\nabla U(x+vt),\quad \lambda^\varepsilon_t:=V^\varepsilon(t)^{\top}\nabla U(X^\varepsilon(t)). 
$$
For $t\ge 0$ define the jump counter 
\[
   N_t^{\varepsilon}
   \;:=\;
   \#\bigl\{\,s\in[0,\,t\wedge\tau^{\varepsilon}(x;\gamma)]:\,
             \ V^{\varepsilon}(s)\neq V^{\varepsilon}(s-)\bigr\}.
\]
$N_t^{\varepsilon}$ counts the discontinuities of the velocity process $V^{\varepsilon}$ that occur no later than time $t$ or the exit time $\tau^{\varepsilon}(x;\gamma)$. This equals the number of PDMP events up to $t$ or exit time.


\begin{proposition}\label{prop:bps_constraint}
For $0<\gamma<1$ and assume that $x^\varepsilon\in B_{\gamma}$ and $v^\varepsilon\neq 0$.  
There exists constant $C=C_\gamma>0$ such that for $T\ge 1$, 
$$\E\left[N_T^\varepsilon\right]\le C~T~(1+|v^\varepsilon|)~\varepsilon^{-1/2}. 
$$
In addition, if $\lambda_0^\varepsilon=0$, 
$$
\E\left[\int_0^{T\wedge\tau^\varepsilon}|\lambda^\varepsilon_s|\dif s\right]\le C~T~(1+|v^\varepsilon|)~\varepsilon^{1/2}. 
$$
\end{proposition}

\begin{proof}
Observe that 
$$
\E\left[N_T^\varepsilon\right]=\varepsilon^{-1}\E\left[\int_0^{T\wedge\tau^\varepsilon}(\lambda^\varepsilon_s)_+\dif s\right]. 
$$
From Dynkin's formula, for each fixed $0 \le t \le T$ we have
\[
   \E\left[\lambda^\varepsilon_{T\wedge\tau^\varepsilon}\right] - \E\left[\lambda^\varepsilon_0\right]
   \;=\; \E\Bigl[\int_0^{T\wedge\tau^\varepsilon}\nabla^2 U\bigl(X^\varepsilon(s)\bigr)
         \bigl(V^\varepsilon(s)\bigr)^{\otimes 2}\dif s\Bigr]
      \;-\;2\,\varepsilon^{-1}\E\left[\int_0^{T\wedge\tau^\varepsilon} (\lambda^\varepsilon_s)_+^2\dif s\right].
\]
By isolating the last term, this equation implies
\begin{align*}
\E\left[\int_0^{T\wedge\tau^\varepsilon}(\lambda^\varepsilon_s)_+^2\dif s\right]&=
2\varepsilon~\left(\E\left[\lambda^\varepsilon_0\right]-\E\left[\lambda^\varepsilon_{T\wedge\tau^\varepsilon}\right]+ \E\Bigl[\int_0^{T\wedge\tau^\varepsilon}\nabla^2 U\bigl(X^\varepsilon(s)\bigr)
         \bigl(V^\varepsilon(s)\bigr)^{\otimes 2}\dif s\Bigr]\right)
\\
&\le 2\varepsilon ~(2|v^\varepsilon|~\overline{\eta}+T~\overline{\kappa}~|v^\varepsilon|^2), 
\end{align*}
where we used $|V^\varepsilon(s)|=|v^\varepsilon|$. 
Cauchy--Schwarz inequality, we obtain
\begin{align*}
\E\left[\int_0^{T\wedge\tau^\varepsilon}(\lambda^\varepsilon_s)_+\dif s\right]
\le \left\{\E\left[\int_0^{T\wedge\tau^\varepsilon}(\lambda^\varepsilon_s)_+^2\dif s\right]\right\}^{1/2}~T^{1/2}. 
\end{align*}
This completes the first claim. Second bound follows by Lemma \ref{lem:monotone}. 
\end{proof}

Proposition \ref{prop:bps_constraint} provides the upper bound of the expectation of $\varepsilon^{1/2}N_T^\varepsilon$. 
However, the inequality also holds in the opposite direction. 

\begin{corollary}\label{cor:bps_constraint}
    If $x^\varepsilon\in B_{\gamma}$ and   $\lambda_0^\varepsilon=0$, for some $C>0$, 
    $$
    \mathbb{P}\left(\varepsilon^{1/2}N_T^\varepsilon\le M^{-1}\right)\le \frac{\varepsilon^{1/2}+M^{-1}}{T}~|v^\varepsilon|. 
    $$
    for any $M,T>0$. 
\end{corollary}

\begin{proof}
    By Lemma \ref{lem:monotone}, since $\lambda_0^\varepsilon=0$, we know that for a constant $c>0$, 
$$
\frac{c~T^2}{1+N_t^\varepsilon}\le \int_0^t(\lambda^\varepsilon_s)_+\dif s, 
$$
therefore, 
\begin{align*}
\mathbb{P}\left(\varepsilon^{1/2}N_T^\varepsilon\le M^{-1}\right)&\le 
    \mathbb{P}\left(\frac{c~T^2}{1+\varepsilon^{-1/2}M^{-1}}\le\int_0^{T\wedge\tau^\varepsilon}(\lambda_s^\varepsilon)_+\dif s\right)\\
    &\le 
    \frac{1+\varepsilon^{-1/2}M^{-1}}{c~T^2}\E\left[\int_0^{T\wedge\tau^\varepsilon}(\lambda_s^\varepsilon)_+\dif s\right]\\
    &\le \frac{\varepsilon^{1/2}+M^{-1}}{c~T}C~|v^\varepsilon|. 
\end{align*}
\end{proof}

\begin{proof}[Proof of Proposition \ref{prp:no_of_jumps}]
First define the first time at which the intensity reaches zero:
$$
t^{\varepsilon}(z)\;:=\;\inf\bigl\{\,t>0 : \lambda^{\varepsilon}_{t}(z)=0\bigr\}.
$$
\begin{enumerate}
    \item[(i)] Upper bound.  We have the following bound:
$$
\E\!\Bigl[N^{\varepsilon}_{\,T_{1}\wedge \tau^{\varepsilon}}\Bigr]
   \;\le\; 1
   +\E\!\Bigl[
        N^{\varepsilon}_{\,T_{1}\wedge \tau^{\varepsilon}}
        \,\mathbf 1_{\{\,t^{\varepsilon}<T_{1}\wedge \tau^{\varepsilon}\}}
     \Bigr].
$$
If $t^{\varepsilon}>T_{1}$ the process performs no jumps when $\lambda^{\varepsilon}_{0}<0$ and at most one jump when $\lambda^{\varepsilon}_{0}\ge 0$. When $t^\varepsilon<T_1$, we may apply Proposition \ref{prop:bps_constraint}: 
$$
\E[N^\varepsilon_{T_1\wedge\tau^\varepsilon}\mid\mathcal{F}_{t^\varepsilon}^\varepsilon\vee \sigma(T_1)]\le C~T_1~(1+|v^\varepsilon|)~\varepsilon^{-1/2}. 
$$
To bound the indicator term, note that $\dot{\lambda}_t^\varepsilon\le \overline{\kappa}~|v^\varepsilon|^2$. We have a lower bound of the hitting time $t^\varepsilon$:
$$
\underline{S}:=\frac{|\lambda_0^\varepsilon|}{\overline{\kappa}|v^\varepsilon|^2}<t^\varepsilon. 
$$
Combining the two pieces yields
\begin{align*}
 \E[N_{T_1}^\varepsilon]&\le 
    1+C~\varepsilon^{-1/2}~\E\left[T_1(1+|v^\varepsilon|)~\mathbf 1_{\{\underline{S}\le T_1\}}\right].  
\end{align*}
For the fixed $\rho$ the right-hand side is $O(\varepsilon^{-1/2})$. Next let $\rho\rightarrow\infty$. Observe 
$$
|\lambda_0^\varepsilon|=|\nabla U(x^\varepsilon)|~|v^\varepsilon|~|U_1|,\quad U_1:=(v^\varepsilon/|v^\varepsilon|)^{\top}n(x^\varepsilon). 
$$ 
$U_1$ is the first coordinate of a vector drawn uniformly from the unit $N$-sphere. Hence $U_1^2\sim B(1/2, (N-1)/2)$, so 
\begin{align*}
    \mathbb{P}(\underline{S}\le T_1\mid~T_1, |v^\varepsilon|)&\le 
\mathbb{P}\left(\frac{\underline{\eta}~|v^\varepsilon|~|U_1|}{\overline{\kappa}|v^\varepsilon|^2}\le T_1\mid~T_1, |v^\varepsilon|\right)\\
&\le\mathbb{P}\left(~|U_1|\le \underline{\eta}^{-1}\overline{\kappa}|v^\varepsilon|~T_1\mid~T_1, |v^\varepsilon|\right)\\
&\le C~|v^\varepsilon|~T_1
\end{align*}
for some constant $C>0$, since $\mathbb{P}(U_1^2\le x)\le C \sqrt{x}$ for some $C>0$. This proves that $\E[N_{T_1}^\varepsilon]$ is at most on the order of $1+\varepsilon^{-1/2}\rho^{-2}$. 
\item[(ii)]
Lower bound. We have
\begin{align*}
    \E[N^\varepsilon_{T_1\wedge\tau^\varepsilon}\mid\mathcal{F}_{t^\varepsilon}^\varepsilon\vee \sigma(T_1)]&\ge 
    M^{-1}\left(1-\mathbb{P}\left(N^\varepsilon_{T_1\wedge\tau^\varepsilon}\le M^{-1}\mid\mathcal{F}_{t^\varepsilon}^\varepsilon\vee \sigma(T_1)\right)\right)\\
    &
    \ge 
    M^{-1}\left(1-\frac{\varepsilon^{1/2}+M^{-1}}{T_1-t^\varepsilon}|v^\varepsilon|\right)_+
\end{align*}
Choose $M$ so that the brackets equals $1/2$; we obtain the lower bound
$$
\left(\frac{T_1-t^\varepsilon}{2~|v^\varepsilon|}-\varepsilon^{1/2}\right)_+. 
$$
Since $ \underline{\kappa}|v^\varepsilon|^2\le \dot{\lambda}_t^\varepsilon\le \overline{\kappa}~|v^\varepsilon|^2$, if $\lambda_0^\varepsilon<0$, 
we have the upper bound of the hitting time $t^\varepsilon$, and the lower bound of the hitting time $\tau^\varepsilon$
$$
t^\varepsilon< \frac{-\lambda_0^\varepsilon}{\underline{\kappa}|v^\varepsilon|^2}=:\overline{S},\quad 
\frac{\gamma}{\overline{\kappa}|v^\varepsilon|^2}<\tau^\varepsilon. 
$$
Consider the following event: 
$$
A:=\left\{T_1>2\overline{S},\ 1/2<|v^\varepsilon|<1,\ T_1<\tau^\varepsilon,\ U_1<0\right\}. 
$$
Then 
\begin{align*}
    \E[N^\varepsilon_{T_1\wedge\tau^\varepsilon}]
    &\ge
    \E\left[\left(\frac{T_1-t^\varepsilon}{2~|v^\varepsilon|}-\varepsilon^{1/2}\right)_+,~ A\right]\\
&\ge
    \E\left[\left(\frac{T_1}{4}-\varepsilon^{1/2}\right)_+,~ A\right]. 
\end{align*}
Arguing as in part (i) and using the tail estimate for $U_1$, we obtain the lower bound. 
\end{enumerate}
\end{proof}

\subsection{Proof of Theorem \ref{thm:bps_refresh_trajectory}}\label{app:subsec:bps_refresh_low}

Let $(X^\varepsilon(t), V^\varepsilon(t))$ denote the Bouncy Particle Sampler process starting from $z=(x, v)$. Define 
$$
V_1^\varepsilon(t)=n^\perp(X^\varepsilon(t))V^\varepsilon(t)
$$
to be the tangential component of $V^\varepsilon$, and define
$$
\phi_t^\varepsilon(z)=(X^\varepsilon(t), V_1^\varepsilon(t)).
$$
We first analyse the convergence of the process to a deterministic flow $\phi_t(z)$ defined in (\ref{eq:bpsflow}) in the absence of refreshment jumps, i.e., $\rho =0$.

\begin{proposition}\label{prop:bps}
Suppose that $\rho =0$. 
Let $z=(x,v)$ and $z^\varepsilon=(x^\varepsilon, v^\varepsilon)$ be  the points in the tangency set with $x^*\neq x$, and suppose $z^\varepsilon\rightarrow z$.  Then, for any $T>0$, the following convergence holds in probability: 
$$
\sup_{0\le t\le T}|\phi_t^\varepsilon(z^\varepsilon)-\phi_t(z)|\longrightarrow 0. 
$$
\end{proposition}

\begin{proof}
Choose $\gamma$ small enough so that $x\in B_{\gamma}$. 
First we consider a convergence localised to the hitting time to high-probability region $\tau^\varepsilon(x^\varepsilon;\gamma)$. Without loss of generality, we can assume that not only $X^\varepsilon(t)$, but also $x(t)$ is in $B_{\gamma} $. This is because $U(x(t))$ does not change. 
Let 
$$
 \alpha(T):=\sup_{0\le t\le T}|X^\varepsilon(t)-x(t)|,\quad 
 \beta(T):=\sup_{0\le t\le T}|V_1^\varepsilon(t)-v(t)|. 
$$
By assumption, it follows that both $\alpha(0)$ and $\beta(0)$ converge to zero. 
We have 
\begin{align}\label{eq:alpha_bound}
    \alpha(t\wedge \tau^\varepsilon)-\alpha(0)
    \le \int_0^{t\wedge\tau^\varepsilon}|V^\varepsilon(s) - v(s)|\dif s\le \int_0^{t\wedge \tau^\varepsilon}\beta(s)\dif s + \int_0^{t\wedge\tau^\varepsilon}|V^\varepsilon(s)-V_1^\varepsilon(s)|\dif s.
\end{align}
The second term on the right-hand side is negligible by Proposition \ref{prop:bps_constraint} since the integrand satisfies 
$$
|V^\varepsilon-V^\varepsilon_1|=|n(X^\varepsilon)^{\top}V^\varepsilon|\le \underline{\eta}^{-1}~|\lambda^\varepsilon|, 
$$
and the right-hand side converges to zero by the proposition. 
We now bound $\beta(t)$. 
Define the error term $\Psi^\varepsilon$ so that 
\begin{align*}
    \dot{V}^\varepsilon_1(t)=\Xi(X^\varepsilon(t),V^\varepsilon(t))+\Psi^\varepsilon(t). 
\end{align*}
The error term has an expression
\begin{align*}
    \Psi^\varepsilon(t)
    &= \psi\bigl(X^\varepsilon(t), V^\varepsilon(t)\bigr),
\end{align*}
where
$$
\psi(x,v)
    := -\,\frac{v^\top\,\nabla U(x)}{|\nabla U(x)|^2}v^\top\,\nabla^2 U(x)
    + 2\,\frac{v^\top\,\nabla U(x)}{|\nabla U(x)|^4}\nabla^2 U(x)\bigl[v,\nabla U(x)\bigr]\,\nabla U(x)
$$
Since we have uniform bounds (\ref{eq:local_uniform_bound}) for the gradient and Hessian on the band $B_\gamma$, we have
\begin{align*}
|\psi(x,v)|\le (\underline{\eta}^2~\overline{\kappa}~|v|+2\underline{\eta}^{-4}~\overline{\eta}^2~\overline{\kappa}~|v|)~|v^{\top}\nabla U(x)|. 
\end{align*}
Recall that the length of $v$ is preserved by the dynamics. 
Thus the contribution from $\Psi^\varepsilon$ is negligible by Proposition \ref{prop:bps_constraint}. 
On the other hand, on the band $B_\gamma$, we have a bound
\begin{align*}
    |\Xi(x,v)-\Xi(x',v')|\le C~|v|^2~|x-x'|+C~|v|~|v-v'|+C~|v'|~|v-v'|. 
\end{align*}
This implies that there exists a constant $C>0$ such that
\begin{align*}
\beta(t\wedge \tau^\varepsilon)-\beta(0)&\le \int_0^{t\wedge\tau^\varepsilon}|\dot{V}_1^\varepsilon(s)-\dot{v}(s)|\dif s\\
&\le \int_0^{t\wedge\tau^\varepsilon}|\Xi(X^\varepsilon(t),V^\varepsilon(t))-\Xi(x(t),v(t))|\dif s+\int_0^{t\wedge\tau^\varepsilon}|\Psi^\varepsilon(s)|\dif s\\
&\le C\int_0^{t\wedge\tau^\varepsilon}\alpha(s)\dif s+C\int_0^{t\wedge\tau^\varepsilon}\beta(s)\dif s+o_p(1). 
\end{align*}
Combining this with (\ref{eq:alpha_bound}) and applying Gronwall’s inequality, we obtain the convergence of $\beta(T\wedge\tau^\varepsilon)$.  
It then follows directly from (\ref{eq:alpha_bound}) that $\alpha(T\wedge \tau^\varepsilon)$ converges in probability.

Finally, we remove the stopping-time condition. In the limit process (\ref{eq:bpsflow}) the trajectory remains within $B_{\gamma}$ because the potential is preserved. Since $X^\varepsilon$ converges to this limit process, it also remains in the band, and hence the stopping-time condition is unnecessary.
\end{proof}

Next, we continue to assume $\rho = 0$, but no longer require $z$ or $z^*$ to lie in the tangency sets.
Under this change, the process undergoes a free-flight period until it reaches the tangency set.
Once it hits the tangency set, its dynamics are governed by Proposition \ref{prop:bps}.

\begin{corollary}
Assume $\rho =0$. 
Suppose that $z^\varepsilon=(x^\varepsilon, v^\varepsilon)$ converges to $z=(x,v)$  where $x^*\neq x$.  For $T>0$, the following convergence holds in probability: 
$$
    \sup_{0\le t\le T}|\Phi_t^\varepsilon(z^\varepsilon)-\Phi_t(z)|\longrightarrow 0. 
$$
\end{corollary}

\begin{proof}
As in the preceding proposition, we again localise time; to keep the notational simplicity, we suppress the stopping time $\tau^{\varepsilon}(x^{\varepsilon};\gamma)$.  After localising time,  the two trajectories 
$$
t\mapsto \Phi_t^\varepsilon(z^\varepsilon),\quad t\mapsto \Phi_t(z)
$$ are both Lipschitz continuous with the same constant $L>0$. Set 
\begin{align*}
    t(z)=\inf\{t>0: \lambda_t(z)=0\}=\inf\{t>0: \lambda(\Phi_t(z))=0\}. 
\end{align*}
Since $\lambda(\Phi_z(t))$ is strictly increasing and continuous, the map $z\mapsto t(z)$ is continuous by the implicit function theorem. 
We split into three cases according to the initial angular momentum $\lambda=\lambda_0(z)=W_0^{\top}\nabla U(x)$. 

\begin{enumerate}
    \item[(i)] 
If $\lambda<0$, then both flows stay in free-flight until they hit the tangency set at times $t(z^\varepsilon)$ and $t(z)$, respectively. By continuity of $t(z)$, we have $t(z^\varepsilon)\rightarrow t(z)$. Also, we have convergence of $\Phi_t(z^\varepsilon)$ to $\Phi_t(z)$ during the free-flight period.  Hence, on the interval $[0, t(z^\varepsilon)]$, 
$$
\sup_{0\le t\le t(z^\varepsilon)}|\Phi_t^\varepsilon(z^\varepsilon)-\Phi_t(z)|\le 
\sup_{0\le t\le t(z)\wedge t(z^\varepsilon)}|\Phi_t(z^\varepsilon)-\Phi_t(z)|+|t(z)-t(z^\varepsilon)|~L\rightarrow 0. 
$$
The proof is done for $T\le t(z^\varepsilon)$ case. For $T>t(z^\varepsilon)$ case, 
for the remainder of the interval we shift time:
\begin{align*}
\sup_{t(z^\varepsilon)\le t\le T}|\Phi_t^\varepsilon(z^\varepsilon)-\Phi_t(z)|&\le
\sup_{0\le s\le T-t(z^\varepsilon)}|\Phi_{s+t(z^\varepsilon)}^\varepsilon(z^\varepsilon)-\Phi_{s+t(z)}(z)|\\
&\le 
\sup_{0\le s\le T-t(z^\varepsilon)}|\Phi^\varepsilon_{s+t(z^\varepsilon)}(z^\varepsilon)-\Phi_{s+t(z)}(z)|+|t(z)-t(z^\varepsilon)|~L\\
&\le 
\sup_{0\le s\le T-t(z^\varepsilon)}|\phi_s^\varepsilon(\Phi_{t(z^\varepsilon)}(z^\varepsilon))-\phi_s(\Phi_{t(z)}(z))|+|t(z)-t(z^\varepsilon)|~L.   
\end{align*}
Proposition \ref{prop:bps} now gives the desired convergence, as 
$\Phi_{t(z^\varepsilon)}(z^\varepsilon)\rightarrow \Phi_{t(z)}(z)$. 

\item[(ii)] If $\lambda>0$, a bounce occurs immediately at the stopping time $\sigma^\varepsilon$: 
\begin{align*}
    \mathbb{P}\left(\sigma^\varepsilon>\kappa\right)&=\exp\left(-\varepsilon^{-1}\int_0^\kappa\lambda_t(z^\varepsilon)\dif t\right)\le \exp(-\varepsilon^{-1}\kappa \lambda_0(z^\varepsilon))\longrightarrow 0
\end{align*}
where we used the fact that $t\mapsto\lambda_t(z^\varepsilon)$ is monotonically increasing. 
Shifting time and invoking the previous case yields
$$
\sup_{0\le t\le T}|\Phi_t^\varepsilon(z^\varepsilon)-\Phi_t(z)|\le 
\sup_{0\le t\le T-\sigma^\varepsilon}|\Phi_{t}^\varepsilon(\Phi_{\sigma^\varepsilon}(z^\varepsilon))-\Phi_t(z)|+\sigma^\varepsilon~L
$$
with $\Phi_\sigma(z^\varepsilon)\rightarrow (x, B(x)v)=z'$ and $\Phi_t(z)=\Phi_t(z')$, so the claim follows.

\item[(iii)] The last case $\lambda=0$ can be proved in the same way depending on the sign of $W_0^{\top}\nabla U(X^\varepsilon)$. 
\end{enumerate}
Finally, since the process never hits the boundary, the foregoing estimates remain valid without the localisation, completing the proof.
\end{proof}

\begin{proof}[Proof of Theorem \ref{thm:bps_refresh_trajectory}]
Denote $N_T$ by the total number of refreshment jumps up to time $T$. Then, for any $\delta>0$, we can choose $K$ such that $\mathbb{P}(N_T>K)\le \delta$. Fix refreshment jump times $0<T_1<T_2<\cdots$ in the following argument. 

Define the error over the $k$-th time block of length $T$ by 
$$
e_k^\varepsilon:=\sup_{T_{k-1}< t\le T_k}|\Phi^\varepsilon_t(z^\varepsilon)-\Phi_t(z)|. 
$$
Assume inductively that $e_k^{\varepsilon} \xrightarrow{\;\mathbb P\;} 0$. Then
$e_k^{\varepsilon} \to 0$ implies $\Phi_{T_k}^{\varepsilon}(z^{\varepsilon}) \to \Phi_{T_k}(z)$.
Applying Proposition \ref{prop:bps}, now with $T_{k+1}-T_k$ in place of $T$, yields $e_{k+1}^{\varepsilon} \xrightarrow{\;\mathbb P\;} 0$.
By induction, $e_k^{\varepsilon} \xrightarrow{\;\mathbb P\;} 0$ for every $k \ge 1$.

On the other hand,  we can bound the total deviation by 
\begin{align*}
\mathbb{P}\left(\sup_{0\le t\le T}|\Phi^\varepsilon_t(z^\varepsilon)-\Phi_t(z)|>\delta\right)&\le 
\delta + \mathbb{P}\left(\sup_{0\le t\le T}|\Phi^\varepsilon_t(z^\varepsilon)-\Phi_t(z)|>\delta,\ N_T\le K\right)\\
&\le \delta+\sum_{k=1}^{K+1}\mathbb{P}\left(e_k^\varepsilon>\delta\right). 
\end{align*}
Since each $e_k^\varepsilon$ converges to zero, the claim follows. 
\end{proof}

\subsection{Proof of Theorem \ref{thm:bps_high_refresh}}\label{app:subsec:bps_high_ref}

Rather than analysing $X^\varepsilon$ directly, we work with the pure-jump Markov process
$$
Y^\varepsilon(t)=X^\varepsilon(T_n)\quad\mathrm{for}\ T_n\le t<T_{n+1} 
$$
which allows us to use Theorem IX.4.21 of \cite{jacodshryaev}, convergence of characteristics for pure-jump semimartingales, on the convergence of such processes. 
As described in the proof below, both processes converge to the same limit. 
The generator of the pure-jump process $Y^\varepsilon$ is given by
$$
A^\varepsilon f(x)=\rho~\E[f(X^\varepsilon(T_1))-f(X^\varepsilon(0))\mid X^\varepsilon(0)=x]
$$
with initial velocity $W_0\sim\mathcal{N}(0, I_N)$. In the limit, the evolution becomes deterministic and satisfies the ODE:  $\dot{x}=b(x)$ with $b(x)=-n(x)/\sqrt{2\pi}$. The drift coefficient can also be written as
$$
b(x)=\E\left[\rho~T_1~B_+(x)W_1\right]. 
$$

\begin{proof}[Proof of Theorem \ref{thm:bps_high_refresh}]
First we show that it is sufficient to show that the gap  between $X^\varepsilon$ and $Y^\varepsilon$ is negligible. 
For $T_n\le t< T_{n+1}$, we have
$$
|Y^\varepsilon(t)-X^\varepsilon(t)|=
|X^\varepsilon(T_n)-X^\varepsilon(t)|\le |W_n|~|T_{n+1}-T_n|=:Z_n. 
$$
Since $W_n\sim\mathcal{N}(0, I_N)$ and $(T_{n+1}-T_n)$ follows the exponential distribution with mean $\rho^{-1}$, we have
\begin{equation}\nonumber
    \E[Z_n^2]=2N\rho^{-2}. 
\end{equation}
Consequently, 
\begin{equation}\nonumber
    \sup_{0\le t\le T}|Y^\varepsilon(t)-X^\varepsilon(t)|\le \max_{1\le n\le N_T}Z_n\rightarrow 0
\end{equation}
by Markov's inequality. 
Hence, once the convergence for $Y^\varepsilon(t)$ has been established, Lemma VI.3.31 of \cite{jacodshryaev} immediately yields the desired result for $X^\varepsilon(t)$.

We prove the convergence of $Y^\varepsilon(t)\rightarrow x(t)$ in probability. 
Since the incremental gap $|X^\varepsilon(T_1)-X^\varepsilon(0)|$ is bounded by $Z_1$, we have
$$
c^\varepsilon(x)=\rho~\E[(X^\varepsilon(T_1)-X^\varepsilon(0))^2\mid X^\varepsilon(0)=x]\le 2N\rho^{-1}\longrightarrow 0
$$
local uniformly. Hence it remains to prove the locally uniformly  convergence
$$
b^\varepsilon(x)=\rho~\E[X^\varepsilon(T_1)-X^\varepsilon(0)\mid X^\varepsilon(0)=x]\longrightarrow b(x). 
$$
By the dominated convergence theorem, it suffices to show that 
$$
\xi=\rho~(X^\varepsilon(T_1)-X^\varepsilon(0))-\rho~T_1~B_+(X^\varepsilon(0))W_1\rightarrow 0
$$
for each fixed $W_1$, Here, $X^\varepsilon(0)=x^\varepsilon\rightarrow x\neq x^*$. 
We have
$$
X^\varepsilon(T_1)=
\begin{cases}
    x^\varepsilon+T_1~W_1&\mathrm{if}\ W_1^\top n(x^\varepsilon)\le 0\\
    x^\varepsilon+\sigma~W_1+(T_1-\sigma)~B(x^\varepsilon+\sigma~W_1)W_1
    &\mathrm{if}\ W_1^\top n(x^\varepsilon)>0\\
\end{cases}
$$
where $\sigma$ denotes the first bounce jump. Here, we neglect 
the event that the process hits the tangency set during the time interval $[0, T_1)$, as this probability tends to zero. 
By the above expression, if $W_1^{\top}n(x^\varepsilon)<0$, then $\xi$ vanishes.
Otherwise, $\rho~\sigma\rightarrow 0$ by 
\begin{align*}
    \mathbb{P}\left(\rho~\sigma>\kappa\right)&=\exp\left(-\varepsilon^{-1}\int_0^{\rho^{-1}\kappa}\lambda_t(z^\varepsilon)\dif t\right)\le \exp(-(\varepsilon\rho)^{-1}\kappa \lambda_0(z^\varepsilon))\longrightarrow 0. 
\end{align*}
This proves
$$
\xi =\rho\sigma W_1+\rho T_1 (B(x^\varepsilon+\sigma W_1)-B(x^\varepsilon))W_1-\rho\sigma B(x^\varepsilon+\sigma W_1)W_1\rightarrow 0. 
$$
The theorem follows from Theorem IX.4.21 of \cite{jacodshryaev}. 
\end{proof}

\subsection{Proof of Theorem \ref{thm:zigzag_trajectory}}
\label{app:subsec:zigzag_trajectory}

Throughout the proof, let
$$
\lambda_i(z)=v_i~\partial_i U(x),\quad \lambda_{t,i}^\varepsilon=V_i^\varepsilon(t)~\partial_i U(X^\varepsilon(t)) 
$$
for all $i=1,\ldots, N$. 
For $t\ge 0$ define the coordinate-wise jump counter 
\[
   N_{i,t}^{\varepsilon}
   \;:=\;
   \#\bigl\{\,s\in[0,\,t\wedge\tau^{\varepsilon}(x;\gamma)]:\,
             \ V_i^{\varepsilon}(s)\neq V_i^{\varepsilon}(s-)\bigr\}.
\]
Thanks to the convexity, it holds that 
$$
\underline{\kappa}~I\le H_{KK}(x)\le \overline{\kappa}~I
$$
for constants $0<\underline{\kappa}\le \overline{\kappa}$, 
where $I$ is the identity matrix for the tangency coordinates $K$. 

\begin{corollary}\label{cor:zz_constraint}
For $\gamma\in (0,1)$, suppose that $\nabla_i U(x^\varepsilon)=0$ with  $U(x^*)+\gamma<U(x^\varepsilon)<U(x^*)+\gamma^{-1}$.  
For $T>0$, there exist constants $0<c<C$ such that 
$$c~\varepsilon^{-1/2}\le \E\left[N_{i, T}^\varepsilon\right]\le C~\varepsilon^{-1/2}
$$ 
for all sufficiently small $\varepsilon>0$. 
Moreover, 
$$
\E\left[\sup_{0\le t\le T\wedge\tau^\varepsilon(x; \gamma)}|\nabla_i U(X^\varepsilon(t))|\right]\le C~\varepsilon^{1/4}
$$
\end{corollary}

\begin{proof} We apply the same argument as in the proof for $\lambda^\varepsilon_t$ in Proposition \ref{prop:bps_constraint} and  Corollary \ref{cor:bps_constraint} to the expression $\lambda_{t,i}^\varepsilon$. Specifically, we observe that $\lambda_{t,i}^\varepsilon$ satisfies the conditions outlined in Lemma \ref{lem:monotone}. This follows from the diagonal-dominance assumption, which ensures that the necessary monotonicity condition holds for $\lambda_{t,i}^\varepsilon$. Therefore, by Lemma \ref{lem:monotone}, the claim is valid. 
\end{proof}

For tangency set $K$, let $\phi_t(x;K)$ be the flow described in Section \ref{sec:zigzag_diagonal_dominance}, that is, 
$$
\phi_t(x:K)=x+t~v.
$$
Denote $J=K^c$. Let 
$$
      v_J = (-\operatorname{sgn}\bigl(\partial_j U(x)\bigr))_{j\in J}, 
      \qquad
      v_K = -\,H_{KK}(x)^{-1}\,H_{KJ}(x)\,v_J.
$$

\begin{proposition}\label{prop:zigzag}
    Let $K\neq\emptyset$ be tangency coordinates of $x$. Suppose that $x^\varepsilon\rightarrow x$. Then the $x$-component  of the Zig-Zag Sampler converges to $\phi_t(x, K)$ until the hitting time $\tau(x;K,\gamma)$ where 
    $$
    \tau(x;K,\gamma)=\tau(x;\gamma)\wedge \min_{j\in J}\inf\{t: \partial_jU(\phi_t(x; K))=0\}. 
    $$
\end{proposition}

\begin{proof}
First, note that without loss of generality, we may suppose that the initial velocity $v^\varepsilon=(v^\varepsilon_K, v^\varepsilon_J)$ satisfies
$$
v_j^\varepsilon = -\operatorname{sgn}\!\bigl(\partial_j U(x^\varepsilon)\bigr), \qquad j\in J.
$$
This implies that if $\varepsilon>0$ is small enough, then $v_j^\varepsilon=v_j$. 
If this condition is not met, there will be an immediate jump so that the equation becomes valid at time $T^\varepsilon$, with $T^\varepsilon\to 0$ as $\varepsilon\to 0$. We can therefore restart the analysis at $T^\varepsilon$.  From that moment on the same argument applies, while the displacement between $X_t^\varepsilon$ and $X_{t+T^\varepsilon}^\varepsilon$ is bounded by $\sqrt{N}\,T^\varepsilon$, which converges to $0$.

Let 
$$
\tau^\varepsilon(x^\varepsilon;K,\gamma)=\tau^\varepsilon(x^\varepsilon;\gamma)\wedge \min_{j\in J}\inf\{t: \partial_jU(X^\varepsilon(t))=0\}. 
    $$
For $0\le t\le T\wedge\tau^{\varepsilon}(x^\varepsilon; K, \gamma)$, there are no jumps for  $X^\varepsilon(t)$ for coordinates $j\in J$. Therefore, for this time interval, 
 $V_j^\varepsilon(t)=v_j$, and 
 $X^\varepsilon_j(t)=x^\varepsilon_j+t~v_j$. 
 In the following, we consider arguments under this time interval. 
By construction,  the time derivative of $\nabla U(X^\varepsilon(t))$ for the coordinate set $K$ is given by 
\[
\frac{\dif }{\dif  t}\bigl[\nabla_K U\bigl(X^\varepsilon(t)\bigr)\bigr]
 = H_{KK}(X^\varepsilon(t))~V_K^\varepsilon(t)+H_{KJ}(X^\varepsilon(t))V_J^\varepsilon(t)=H_{KK}\bigl(X^\varepsilon(t)\bigr)\,W^\varepsilon(t),
\]
where
\[
W^\varepsilon(t)
 := V_K^\varepsilon(t) 
   + H_{KK}^{-1}\bigl(X^\varepsilon(t)\bigr)
         \,H_{KJ}\bigl(X^\varepsilon(t)\bigr)\,V_J^\varepsilon(t).
\]
Let $\Psi^\varepsilon(t)$ be a matrix-valued function defined by
\[
\Psi^\varepsilon(t)[a,b]
 = \frac{\dif }{\dif t}\bigl[H^{-1}_{KK}(X^\varepsilon(t))[a,b]\bigr]
 = -T_K(X^\varepsilon(t))[H_{KK}^{-2}(X^\varepsilon(t))
       V^\varepsilon(t), a, b]
\]
where $T_K$ is the third derivative of $U$ with respect to the coordinate set $K$. On the band $B_\gamma$, we have a uniform bound
$$
\sup_{0\le t\le T\wedge\tau^{\varepsilon}(x; K, \gamma)}|\Psi^\varepsilon(t)|<\infty. 
$$
By the integration by parts formula, we obtain
\[
\int_s^t W^\varepsilon(u)\,\dif u
 = 
   \Bigl[
     H_{KK}^{-1}(X^\varepsilon(u))\,
     \nabla_K U(X^\varepsilon(u))
   \Bigr]_s^t
   - 
   \int_s^t 
     \nabla_K U\bigl(X^\varepsilon(u)\bigr)^{\top}\Psi^\varepsilon(u)
   \,\dif u.
\]
Let $I_1$ and $I_2$ be the first and second term in the right-hand side for $s=0$ and $t= T\wedge \tau^{\varepsilon}(x^\varepsilon; K, \gamma)$. 
We have bounds
$$
|I_1|\le 2~\underline{\kappa}^{-1}~\sum_{k\in K}\sup_{0\le t\le T\wedge \tau^{\varepsilon}(x^\varepsilon; K, \gamma)}|\partial_kU(X^\varepsilon(t))|
$$
and
$$
|I_2|\le C_{\Psi^\varepsilon}~\sum_{k\in K}\int_0^{T\wedge\tau^{\varepsilon}(x^\varepsilon; K, \gamma)}|\partial_kU(X^\varepsilon(t))|\dif t
$$
and both of which converge to $0$ in probability,  
by Corollary \ref{cor:zz_constraint}. 
Next we express $X_K^\varepsilon=(X_i^\varepsilon)_{i\in K}$ as 
\[
\begin{aligned}
X_K^\varepsilon(t)
 &= x_K(0) + \int_0^t V_K^\varepsilon(u)\,\dif u \\[6pt]
 &= 
   x_K(0)
   - 
   \int_0^t 
     H_{KK}^{-1}(X^\varepsilon(u))\,
     H_{KJ}\bigl(X^\varepsilon(u)\bigr)\,V_J^\varepsilon(u)
   \,\dif u
   +
   \int_0^t W^\varepsilon(u)\,\dif u.
\end{aligned}
\]
Since $V_J^\varepsilon(u)=v_J$, we have the following representation: 
\begin{align*}
X_K^\varepsilon(t)-x_K(t)
 = 
   \int_0^t 
     (\psi(X^\varepsilon(u))-\psi(x(u)))
   \,\dif u~v_J
   +
   \int_0^t W^\varepsilon(u)\,\dif u,
\end{align*}
where $\psi(x)=H_{KK}^{-1}(x)H_{KJ}(x)$. 
From this representation, one obtains
\[
\sup_{0\le t\le T\wedge\tau^{\varepsilon}(x^\varepsilon; K, \gamma)}|X_K^\varepsilon(t) - x_K(t)|
 \le 
   C\int_0^{T\wedge\tau^{\varepsilon}(x^\varepsilon; K, \gamma)}|X_K^\varepsilon(t) - x_K(t)|\dif t
   + C~|x_K^\varepsilon-x_K|+
   |I_1|+|I_2|. 
\]
Applying Gronwall’s inequality to the above estimate completes the proof of the claim. Finally, $\tau^{\varepsilon}(x^\varepsilon; K, \gamma)$ converges to $\tau(x; K, \gamma)$ and so we can replace the former by the latter. The claim holds by taking $T\rightarrow\infty$. 
\end{proof}

\begin{proof}[Proof of Theorem \ref{thm:zigzag_trajectory}]
Let 
\[
0 =T_0< T_1 < T_2 < \cdots <T_N
\]
be the hitting times of the limit process to one of the boundaries $\{x: \partial_iU(x)=0\}$. 
Applying Proposition \ref{prop:zigzag} to each segment up to $\tau^\varepsilon(x^\varepsilon;\gamma)$, the claim follows. 
\end{proof}

\section{Analysis for the jump chains}\label{app:sec:fec}

Throughout this section, assume  Assumption \ref{assumption}.
Define $\lambda_t = \lambda_t(x, v) := v^\top \nabla U(x + vt)$, and let $\Lambda_t := U(x + vt) - U(x)$, and write $n(x)=\nabla U(x)/|\nabla U(x)|$. Define the tangency time as 
\begin{equation}\label{eq:tangencytime}
t(x,v)=\inf\{t> 0: \lambda_t=0\}.
\end{equation}
The tangency time $t(x,v)$ will be used as a reference against the PDMP event time with survival function given by: 
\begin{equation}\label{eq:jump_time}
\mathbb{P}(t^\varepsilon(x,v)> t)=\exp\left(-\varepsilon^{-1}\Lambda_t^+(x,v)\right). 
\end{equation}
where
\begin{align*}
    \Lambda_t^+(x,v)=\int_0^t\lambda_s^+(x,v)\dif s,\quad \mathrm{and}\quad \lambda_t^+(x,v)=\max\{\lambda_t(x,v),0\}. 
\end{align*}
By Assumption \ref{assumption},  $\Lambda_\infty^+(x,v)=\lim_{t\rightarrow\infty}\Lambda_t^+(x,v)=+\infty$ as 
$$
\Lambda_t^+(x,v)\ge \Lambda_t(x,v)= U(x+vt)-U(x)\rightarrow +\infty\quad (t\rightarrow +\infty). 
$$
For $\varepsilon>0$, consider the following Markov kernel 
\begin{equation}\label{eq:kernel}
P^\varepsilon(x, A)=\int_{\mathbb{R}^N}\int_0^\infty \mathbf 1_A(x+vt)~\varepsilon^{-1}~\lambda_t^+(x,v)\exp(-\varepsilon^{-1}~\Lambda_t^+(x,v))\dif t~Q_x(\dif v)
\end{equation}
where $Q_x(\dif v)$ is a probability measure on $\mathbb{R}^N\backslash\{0\}$ such that $\lambda_0(x,v)<0$ with probability $1$. 

\begin{proposition}\label{prop:drift}
For $V(x)=\exp(U(x))$ and for $\varepsilon\in (0,1)$, we have
$$
\frac{P^\varepsilon V(x)-V(x)}{V(x)}\le \frac{1}{1-\varepsilon}\int_{\mathbb{R}^N}\exp(\Lambda_{t(x,v)})Q_x(\dif v)-1\quad (x\neq x^*).  
$$
Also, when $\gamma>0$, for $A_\gamma(x)=\{y: U(y)\ge U(x)+\gamma\}$ we have
\begin{align*}
    P^\varepsilon(x, A_\gamma(x))\le 2~\exp(-\varepsilon^{-1}\gamma/2)\quad (x\neq x^*). 
\end{align*}
\end{proposition}

\begin{proof}
Since $U(x+vt)-U(x)=\Lambda_t(x,v)$, we have
\begin{align*}
    \frac{P^\varepsilon V(x)-V(x)}{V(x)}=\int_{\mathbb{R}^N}\int_0^\infty\exp\left(\Lambda_t\right)\varepsilon^{-1}\lambda_t^+\exp\left(-\varepsilon^{-1}\Lambda_t^+\right)\dif t~Q_x(\dif v)-1. 
\end{align*}
By strong convexity, 
$\dot{\lambda}_t=\nabla^2U(x+vt)[v^{\otimes 2}]>0$. Therefore, 
$\lambda_t$ is a strictly increasing function, and $\lambda_t^+=0$ for  $t\le t(x,v)$. By the fact,  for $t>t(x,v)$, we have a simple decomposition
$$
\Lambda_t=\Lambda_{t(x,v)}+\Lambda_t^+. 
$$ 
Also, we have 
\begin{align*}
\int_0^\infty\exp\left(\Lambda_t^+\right)\varepsilon^{-1}\lambda_t^+\exp\left(-\varepsilon^{-1}\Lambda_t^+\right)\dif t
=
\frac{1}{1-\varepsilon}\int_0^\infty(\varepsilon^{-1}-1)\lambda_t^+\exp\left(-(\varepsilon^{-1}-1)\Lambda_t^+\right)\dif t
\le \frac{1}{1-\varepsilon}, 
\end{align*}
where the last inequality is equality if $\Lambda_\infty^+:=\lim_{t\rightarrow\infty}\Lambda_t=\infty$. 
By the decomposition of $\Lambda_t$, we have
\begin{align*}
\int_0^\infty\exp\left(\Lambda_t\right)\varepsilon^{-1}\lambda_t^+\exp\left(-\varepsilon^{-1}\Lambda_t^+\right)\dif t
&=
\exp(\Lambda_{t(x,v)})~\int_0^\infty\exp\left(\Lambda_t^+\right)\varepsilon^{-1}\lambda_t^+\exp\left(-\varepsilon^{-1}\Lambda_t^+\right)\dif t\\
&\le \frac{1}{1-\varepsilon}\exp(\Lambda_{t(x,v)}). 
\end{align*}
The second claim can be proved in the same way. For $x+vt\in A_\gamma(x)$, $\Lambda_t\ge \gamma$ and hence $\Lambda_t^+\ge\gamma$. In particular, 
$\Lambda_t^+\ge \gamma/2+\Lambda_t^+/2$. Therefore, 
\begin{align*}
    P^\varepsilon(x, A_\gamma(x))&=
\int_{\mathbb{R}^N}\int_0^\infty \mathbf 1_{A_\gamma(x)}(x+vt)\varepsilon^{-1}\lambda_t^+\exp\left(-\varepsilon^{-1}\Lambda_t^+\right)\dif t~Q_x(\dif v)\\
&\le 2~\exp(-\varepsilon^{-1}\gamma/2)~
\int_{\mathbb{R}^N}\int_0^\infty~(\varepsilon^{-1}/2)\lambda_t^+\exp\left(-\varepsilon^{-1}/2~\Lambda_t^+\right)\dif t~Q_x(\dif v)\\
&\le 2~\exp(-\varepsilon^{-1}\gamma/2). 
\end{align*}
\end{proof}

Let $0<\underline{\kappa}\le\overline{\kappa}$ and $0<\underline{\eta}\le\overline{\eta}$ be such that the bounds from Appendix~\ref{app:sec:proof_PDMP} hold uniformly on the band $B_{\gamma/2}$ (in place of $B_{\gamma}$).

\begin{lemma}\label{lem:drift_bound}
For $x\in B_{\gamma}$, if $\lambda_0(x,v)<0$, then for $0\le t\le t(x,v)$ we have
$$
\Lambda_{t}\le -\min\left\{\frac{\underline{\kappa}}{2}|v|^2~t^2,\frac{\gamma}{2}\right\}. 
$$
In particular, we have
\begin{align*}
\Lambda_{t(x,v)}\le -\min\left\{\frac{\underline{\kappa}}{2\overline{\kappa}^2}~\underline{\eta}^2\frac{(v^{\top}n(x))^2}{|v|^2}, \frac{\gamma}{2}\right\}.
\end{align*}
\end{lemma}

\begin{proof}
If $\lambda_0(x,v)<0$, then the value of $U(x)$ is monotonically decreasing until the hitting time $t(x,v)$. The process may go out of $B_{\gamma/2}$ but in that case, $U$ has already decreased by at least $\gamma/2$, so $\Lambda\le -\gamma/2$. 
We consider the other cases, that is, the process is still on $B_{\gamma/2}$. In this case,  we have
$$
0=\lambda_{t(x,v)}=\lambda_0+\int_0^{t(x,v)}\dot{\lambda}_t\dif t 
$$
we have
\begin{equation}\label{eq:stoppingtimebound}
\frac{-\lambda_0}{\overline{\kappa}|v|^2}\le t(x,v)
\end{equation}
since $\dot{\lambda}_t=\nabla^2U(x+vt)[v^{\otimes 2}]$ is bounded below by $\underline{\kappa}|v|^2$. On the other hand, for $t\in [0, t(x,v)]$, we have 
$$
\lambda_t=\lambda_t-\lambda_{t(x,v)}=-\int_t^{t(x,v)}\dot{\lambda}_s\dif s\le -\underline{\kappa}|v|^2(t(x,v)-t). 
$$
Combining these fact, we obtain that
$$
\Lambda_t=\int_0^t\lambda_s\dif s
\le -\int_0^t \underline{\kappa}|v|^2(t(x,v)-s)~\dif s
\le -\frac{\underline{\kappa}}{2}~|v|^2~t^2\quad\Longrightarrow\quad \Lambda_{t(x,v)}\le -\frac{\underline{\kappa}}{2\overline{\kappa}^2}\frac{\lambda_0^2}{|v|^2}, 
$$
where the right-hand side comes from (\ref{eq:stoppingtimebound}). 
Then the last claim follows as $(-\lambda_0)=|v^{\top}\nabla U(x)|\ge \underline{\eta}~|v^{\top}n(x)|$. 
\end{proof}

\subsection{Proof of Theorem \ref{thm:bps_refresh}}\label{app:subsec:bps_jump_chain}
We now define the one-step refresh kernel $P$ of the ODE-with-jumps limit (see Subsection~\ref{subsec:lowrefreshbps}). Let $\rho>0$ be the refresh rate. Fix $x\in\mathbb{R}^N\backslash\{x^*\}$ and set $X_0=x$. Draw
\[
W_1\sim \mathcal{N}(0,I_N),\qquad T_1\sim \mathcal{E}(\rho)=:\Gamma,
\]
independently, where $\mathcal{E}(\rho)$ is the exponential distribution with mean $\rho^{-1}$. Define the reflected initial velocity as  $V_0=B_+(x)W_1$. Let $(X_t,V_t)$ denote the deterministic trajectory that follows the free-flight (\ref{eq:freeflight})
up to the tangency time $t(x,v)$ defined in (\ref{eq:tangencytime}), and then evolves according to the normalised gradient flow (\ref{eq:bpsflow}) after $t(x,v)$. We then set
\[
P(x,A)\;:=\;\mathbb{P}\bigl(X_{T_1}\in A\bigr),\qquad A\in\mathcal{B}(\mathbb{R}^N).
\]
Note that $X_t$ is $U$-non-increasing.

\begin{proposition}\label{prop:bpsdrift}
    Let $V(x)=\exp(U(x))$, and let $\gamma\in (0,1)$. For the one-step refresh kernel $P$, there exists $\beta >0$ and $c >0$ such that, 
    $$
    PV(x)-V(x)\le -\beta  V(x)\quad\mathrm{for}\ x\in B_{\gamma}. 
    $$
\end{proposition}

\begin{proof}
    Let $\Gamma$ be the exponential distribution with parameter $\rho>0$ and let 
\begin{equation}\nonumber
P'(x, A)=\int_{\mathbb{R}^N}\int_0^\infty \mathbf 1_A(x+v~(t(x,v)\wedge T))~~\Gamma(\dif T)~Q_x(\dif v). 
\end{equation}
   The kernel $P$ differs from $P'$ because, under $P$, the process continues to evolve after the tangency time $t(x,v)$, whilst $P'$  freezes the dynamics at that stopping time. As the potential function remains unchanged beyond the hitting time, $P$ and $P'$ agree on $V(x)$. We have
   $$
   \frac{P'V(x)-V(x)}{V(x)}=\int_{\mathbb{R}^N}\int_0^\infty\exp(\Lambda_{t(x,v)\wedge T})\Gamma(\dif T)Q_x(\dif v)-1. 
   $$
    From Lemma \ref{lem:drift_bound}, we have
    $$
    \Lambda_{t(x,v)\wedge s}
    \le -\min\left\{\frac{\underline{\kappa}}{2}|v|^2~s^2,~\frac{\underline{\kappa}}{2\overline{\kappa}^2}~\underline{\eta}^2\frac{(v^{\top}n(x))^2}{|v|^2}\right\}
    $$
    with constants calibrated on $B_{\gamma/2}$. 
    In particular, the right-hand side is negative and the distribution does not depend on $x$, as the joint law of $((v^{\top}n(x), |v|)$ is $x$-free. 
    Thus 
    \begin{equation}\nonumber
    \frac{P'V(x)-V(x)}{V(x)}\le  -\beta \quad (x\in B_{\gamma/2})
    \end{equation}
    where 
    $$
    \beta=1-\E\left[\exp\left(\min\left\{\frac{\underline{\kappa}}{2}|v|^2~T^2,~\frac{\underline{\kappa}}{2\overline{\kappa}^2}~\underline{\eta}^2\frac{(v^{\top}n(x))^2}{|v|^2}\right\}\right)\right]>0
    $$
    with $v\sim Q_x(\dif v)$ and $T\sim\Gamma(\dif T)$ independently. 
 \end{proof}

\begin{proof}[Proof of Theorem \ref{thm:bps_refresh}]
By assumption $x\in L_\gamma^c$. Shrinking $\gamma$ if necessary, assume $x\in B_\gamma\subset L_\gamma^{\,c}$. Let
\[
m(x;\gamma):=\inf\{n\ge 0:\ X_{T_n}\in B_\gamma^c\}.
\]
Because the dynamics are $U$-non-increasing, $m(x;\gamma)$ coincides with the hitting time of the chain to $L_\gamma$  denoted by $n(x;\gamma)$. 

By the drift condition in Proposition~\ref{prop:bpsdrift} there exist $r>1$ and $C<\infty$ such that
\[
\sup_{x\in B_\gamma}\ \E_x\!\left[\sum_{n=0}^{n(x:\gamma)-1} r^n\,V(X_{T_n})\right]=\sup_{x\in B_\gamma}\ \E_x\!\left[\sum_{n=0}^{m(x:\gamma)-1} r^n\,V(X_{T_n})\right]\le\ C.
\]
Since $x^*$ is the minimiser of $U(x)$, we have $V(x)\ge V(x_*)$. Thus 
\[
\sum_{n=0}^{M-1} r^n\,V(X_{T_n})\ \ge\ V(x_*)\,\frac{r^{M}-1}{r-1}. 
\]
Therefore, 
\[
\P_x(n(x;\gamma)\ge M)\ \le\ \frac{C~(r-1)}{V(x^*)}~(r^M-1)^{-1}. 
\]
Given $c>0$, choose $M\ge \lceil \log(1+C(r-1)/cV(x^*))/\log r\rceil$ to obtain $\P_x(\tau_\gamma\ge M)\le c$.

Finally, by Proposition~\ref{prop:bps} the $\varepsilon$-BPS process converges to the limit process $X_t$ as $\varepsilon\to0$, and therefore
\[
\limsup_{\varepsilon\to0}\ \P_x\big(n^\varepsilon(x;\gamma)\ge M\big)
\ \le\ \P_x\big(n(x;\gamma)\ge M\big)\ \le\ c.
\]

\end{proof}

\subsection{Proof of Theorem \ref{thm:fec}}
\label{app:subsec:proof_fec}

Consider the jump chain of the Forward Event-Chain. 
The corresponding Markov kernel $P^\varepsilon$ has the form (\ref{eq:kernel}) where $Q_x(v)$ is the law of 
$$
-\xi n(x)+n^\perp(x)v,
$$
where $\xi\sim\mathrm{Rayleigh}(1)$   and $v\sim\mathcal{N}(0, I_N)$ independently. 

\begin{proposition}\label{prop:fec_dc}
    There exists a constant $\beta >0$ and $\gamma >0$ such that all  sufficiently small $\varepsilon>0$, 
    $$
    P^\varepsilon V(x)-V(x)\le -\beta  V(x)\quad (x\in B_{\gamma}).
    $$
\end{proposition}

\begin{proof}
By Proposition \ref{prop:drift} and Lemma \ref{lem:drift_bound}, if $\varepsilon$ is sufficiently small, we can conclude the inequality holds on $B_\gamma$ with
$$
\beta =1-\frac{1}{1-\varepsilon}~\E\left[\exp\left(-\frac{\underline{\kappa}}{2\overline{\kappa}^2}~\underline{\eta}^2\frac{(v^{\top}n(x))^2}{|v|^2}\right)\right]>0
$$
where the expectation is $x$-free. 
\end{proof}

\begin{proof}[Proof of Theorem \ref{thm:fec}]
Let $X^\varepsilon_t$ be the Forward Event-Chain process, and $T^\varepsilon_n$ be the event times. 
Let
\[
m^\varepsilon(x;\gamma):=\inf\{n\ge 0:\ X^\varepsilon_{T_n^\varepsilon}\in B_\gamma^c\}.
\]

We first prove the claim for $m^\varepsilon(x;\gamma)$ in place of $n^\varepsilon(x;\gamma)$. 
By the exponential drift inequality, we can find a constant $r>1$ such that 
$$
\sup_{x\in B_{\gamma}}\E_x\left[\sum_{n=0}^{m^\varepsilon(x;\gamma)-1}V(X_{T_n^\varepsilon}^\varepsilon)~r^n\right]<C
$$
for some constant $C>0$ which does not depend on $\varepsilon\in (0,\varepsilon')$ for some $\varepsilon'>0$.  
As in the proof of Theorem \ref{thm:bps_refresh}, we obtain 
$$
\mathbb{P}(m^\varepsilon(x;\gamma)\ge M)\le C (r^M-1)^{-1}
$$
for some $C>0$. The right-hand side is smaller than $c>0$
if $M\ge \lceil \log(1+C/c)/\log r\rceil$. 
Thus the probability converges to $0$ as $M\rightarrow\infty$, and the proof for $m^\varepsilon(x;\gamma)$ is completed. 

Now we consider $n^\varepsilon(x;\gamma)$. 
Suppose that $x$ is slightly inside of $B_\gamma$, i.e., 
$x\in B_{\gamma'}$ for some $\gamma'>\gamma$. 
Let 
$$
\delta=\frac{\gamma^{-1}-(\gamma')^{-1}}{M}>0. 
$$
By  the second statement of Proposition \ref{prop:drift}, 
$$
P^\varepsilon(x, A_\delta(x))\le 2~\exp(-\varepsilon^{-1}\delta/2)
$$
for $A_\delta(x)=\{y: U(y)> U(x)+\delta\}$. Thus 
\begin{align*}
\mathbb{P}(X^\varepsilon_{T_n^\varepsilon}\notin L_{\gamma^{-1}}, \exists n=1,\ldots, M-1)&=\mathbb{P}(U(X^\varepsilon_{T_n^\varepsilon})>\gamma^{-1}, \exists n=1,\ldots, M-1)\\
&\le 
\sum_{n=1}^{M-1}\mathbb{P}(X^\varepsilon_{T_n^\varepsilon}\in A_\delta(X^\varepsilon_{T_{n-1}^\varepsilon}))\\
&\le 2(M-1)\exp(-\varepsilon^{-1}\delta/2). 
\end{align*}
Therefore, 
\begin{align*}
\limsup_{\varepsilon\rightarrow 0}\mathbb{P}(m^\varepsilon(x;\gamma)\ge M)&=
\lim\sup_{\varepsilon\rightarrow 0}\mathbb{P}(X^\varepsilon_{T_n^\varepsilon}\in L_\gamma^c \cap L_{\gamma^{-1}}, n=1,\ldots, M-1)\\
&=
\lim\sup_{\varepsilon\rightarrow 0}\mathbb{P}(X^\varepsilon_{T_n^\varepsilon}\in L_\gamma^c, n=1,\ldots, M-1)\\
&\le\limsup_{\varepsilon\rightarrow 0}\mathbb{P}(n^\varepsilon(x;\gamma)\ge M)\le c. 
\end{align*}
\end{proof}

\subsection{Proof of Theorem \ref{thm:cs}}
\label{app:subsec:proof_cs}

Let  $p_n=|\partial_nU(x)|/\sum_{m=1}^N|\partial_m U(x)|$, and let $e_n$ be the $n$-th standard basis vector in $\mathbb{R}^N$
Let $Q_x(\dif v)$ be the law defined by 
picking up index $n$ with probability $p_n$, and then set $v=-s_ne_n$ with $s_n=\operatorname{sgn}(\partial_n U(x))$, i.e., 
$$
Q_x(\dif v)=\sum_{n=1}^Np_n\delta_{-s_ne_n}(\dif v). 
$$
The Markov kernel of the Coordinate Sampler has the form (\ref{eq:kernel}). 

\begin{proposition}\label{prop:cs_dc}
    For $V(x)=\exp(U(x))$,  there exists a constant $\beta >0$ and all sufficiently small $\varepsilon>0$, 
    $$
    P^\varepsilon V(x)-V(x)\le -\beta  V(x)\quad (x\in B_{\gamma})
    $$
   
\end{proposition}

\begin{proof}
Let 
$$
q_n=\frac{|\partial_nU(x)|^2}{\sum_{m=1}^N|\partial_m U(x)|^2}. 
$$
By 
Cauchy--Schwartz inequality, we have
$$
\sum_{n=1}^N|\partial_n U(x)|\le N^{1/2}~\left(\sum_{n=1}^N|\partial_n U(x)|^2\right)^{1/2}. 
$$
Thus 
$$
p_n\ge N^{-1/2}\sqrt{q_n}. 
$$
Pick $n^*=\arg\max q_n$. Since $\sum_{n=1}^Nq_n=1$, we have $q_{n^*}\ge 1/N$, and $p_{n^*}\ge 1/N$. 
By  Lemma \ref{lem:drift_bound}, for $c=(\underline{\kappa}/2\overline{\kappa}^2)~\underline{\eta}^2$, we have
\begin{align*}
\int\exp(\Lambda_{t(x,v)} )Q_x(\dif v)&
       \le 
       \sum_{n=1}^Np_n~\exp\left(-\min\left\{c~q_n,\frac{\gamma}{2}\right\}\right)
\end{align*}
 Thus the above sum is bounded above by 
\begin{align*}
        p_{n^*}\exp\left(-\min\left\{c~q_{n^*},\frac{\gamma}{2}\right\}\right)
        +(1-p_{n^*})\le 1-\frac{1}{N}\left(1-\exp\left(-\min\left\{\frac{c}{N},\frac{\gamma}{2}\right\}\right)\right)<1. 
\end{align*}
By Proposition \ref{prop:drift}, the claim follows. 
\end{proof}

\begin{proof}[Proof of Theorem \ref{thm:cs}]
We omit the proof as it is essentially the same as that of Theorem \ref{thm:fec}.     
\end{proof}


\end{document}